\documentclass[a4paper,envcountsame,envcountsect]{llncs}
\newcommand{\sabove}[2]{\stackrel{\scriptstyle {\text{#2}}}{#1}}
\usepackage{amsmath, amssymb}
\usepackage{multirow}
\usepackage{gastex, float}

\usepackage{color}

\newcommand{\cdp}{\mkern1mu\cdotp}

\newcounter{Lcount}

\newenvironment{conditions}
{%
  \begin{list}{\rm (\theenumi)}%
  {\noindent%
    \usecounter{enumi}%
    \setlength{\topsep}{2pt}%
    \setlength{\partopsep}{0pt}%
                \setlength{\itemsep}{2pt}%
    \setlength{\parsep}{0pt}%
    \setlength{\leftmargin}{2.5em}%
    \setlength{\labelwidth}{1.5em}%
    \setlength{\labelsep}{0.5em}%
    \setlength{\listparindent}{0pt}%
    \setlength{\itemindent}{0pt}%
  }%
}%
{\end{list}}%

\begin{document}

\title{Quantum Finite Automata and Probabilistic Reversible Automata: R-trivial Idempotent Languages
\thanks{Supported by the Latvian Council of
Science, grant No. 09.1570 and by the European Social Fund, contract
No. 2009/0216/1DP/1.1.1.2.0/09/APIA/VIAA/044.}}

\author{
Marats Golovkins
 \and Maksim Kravtsev
 \and Vasilijs Kravcevs}

\institute{Faculty of Computing, University of Latvia, Rai\c na
bulv. 19, Riga LV-1586, Latvia
\\
\email{marats AT latnet DOT lv, maksims DOT kravcevs AT lu DOT lv,\\
kvasilijs AT gmail DOT com}}

\maketitle

\begin{abstract}
We study the recognition of $\mathcal{R}$-trivial idempotent
($\mathcal{R}_1$) languages by various models of "decide-and-halt"
quantum finite automata (QFA) and probabilistic reversible automata
(DH-PRA). We introduce {\em bistochastic} QFA (MM-BQFA), a model
which generalizes both Nayak's enhanced QFA and DH-PRA. We apply
tools from algebraic automata theory and systems of linear
inequalities to give a complete characterization of $\mathcal{R}_1$
languages recognized by all these models. We also find that
"forbidden constructions" known so far do not include all of the
languages that cannot be recognized by measure-many QFA.
\end{abstract}
\section{Introduction}\label{sec:Introduction}
Measure-many quantum finite automata (MM-QFA) were defined in 1997
\cite{KW97} and since then, their language class characterization
problem remains open. The difficulties arise because the language
class is not closed under Boolean operations like union and
intersection \cite{AV00}. Later on, a probabilistic reversible
("decide-and-halt" probabilistic reversible automaton, DH-PRA) and a
more general model of quantum finite automata (enhanced quantum
finite automaton, EQFA) were defined as well, which remarkably share
with MM-QFA the same property of non-closure \cite{GK09,M08}.

Nevertheless, other probabilistic reversible and quantum models of
finite automata are known as well ("classical" probabilistic
reversible automata, C-PRA, and Latvian quantum finite automata,
LQFA), closed under Boolean operations \cite{GK02,AT04}. The
language class characterization problem for these models were solved
by help of algebraic automata theory \cite{AT04}. As a matter of
fact, the language classes of both models form the same language
variety, corresponding to the $\mathbf{EJ}$ monoid variety.

In \cite{AT04}, it is also stated that MM-QFA recognize any regular language
corresponding to the monoid variety $\mathbf{EJ}$. Since any
syntactic monoid of a unary regular language belongs to
$\mathbf{EJ}$, the results in \cite{AT04} imply that MM-QFA
recognize any unary regular language. In \cite{BP10}, MM-QFA
recognizing unary languages are studied in detail, the authors give
a new proof of this result by explicitly constructing MM-QFA recognizing unary languages.

The results by Brodsky and Pippenger \cite{BP99} combined with the
non-closure property imply that the class of languages recognized by
MM-QFA is a proper subclass of the language variety corresponding to
the $\mathbf{ER}$ monoid variety. The same holds for DH-PRA and for
EQFA \cite{GK09,M08}. In the paper, we consider a sub-variety of $\mathbf{ER}$, the variety of
$\mathcal{R}$-trivial idempotent monoids $\mathbf{R_1}$ and
determine which are the $\mathcal{R}$-trivial idempotent languages ($\mathcal{R}_1$ languages)
that are recognizable by DH-PRA, MM-QFA, EQFA and MM-BQFA ("decide-and-halt" models). Since
$\mathbf{R_1}$ shares a lot of the characteristic
properties with $\mathbf{ER}$, the obtained results may serve as an
insight to solve the general problem relevant to $\mathbf{ER}$.

The paper is structured as follows. Section \ref{sec_prelim} gives definitions used throughout the paper. Section \ref{variet} describes
the algebraic tools - monoids, morphisms and varieties. Section \ref{sec_cp_maps} considers completely positive maps. We apply von Neumann-Halperin theorem and the result
by Kuperberg to obtain Theorem \ref{theor_bist_EJ}, which is essential to prove the limitations of QFA in terms of language recognition.
Sections \ref{sec_automata}, \ref{Lin_Ineq}, \ref{construct_DH_PRA}, \ref{section_construct_MM-QFA}, \ref{section_forb_constr} present the main results of the paper:
\begin{conditions}
\item Introduction of MM-BQFA, a model which generalizes the earlier
"decide-and-halt" automata models (Section \ref{sec_automata}, Definition \ref{def_BQFA});
\item Definition of systems of linear inequalities corresponding to $\mathcal{R}_1$ languages. Proof that any $\mathcal{R}_1$ language cannot be
recognized by the "decide-and-halt" models, if its system of linear inequalities is not consistent.
(Section \ref{Lin_Ineq}, Definition \ref{def_inequalities}, Theorem \ref{theor_sys_not_consist});
\item The construction of DH-PRA (this presumes also EQFA and MM-BQFA) and MM-QFA for any $\mathcal{R}_1$ language having a consistent system of inequalities.
Consequently, we obtain that all four "decide-and-halt" models recognize exactly the same $\mathcal{R}_1$ languages. An $\mathcal{R}_1$ language is recognizable
by any of these models if and only if the corresponding system of linear inequalities is consistent.
(Sections \ref{construct_DH_PRA}, \ref{section_construct_MM-QFA}, Theorems \ref{theor_sys_consist}, \ref{theor_sys_consist_QFA}, \ref{theor_iff1}, \ref{theor_iff2});
\item The proof that the "forbidden constructions" known from \cite{AV00} do not give all of the languages that cannot be recognized by MM-QFA (Section \ref{section_forb_constr}, Theorem \ref{theor_forb_constr}).
\end{conditions}
Among other results, we obtain the language class recognized by MO-BQFA (Theorem \ref{theor_EJ_MO-BQFA}) and give some closure properties of MM-BQFA (Theorems \ref{theor_MM_BQFA_subset_ER}, \ref{theor_closed_compl}
and Corollaries \ref{cor_subset}, \ref{cor_not_closed}).
\section{Preliminaries}\label{sec_prelim}
Given an alphabet $A$, let $A^*$ be the set of words over alphabet
$A$. Given a word $\mathbf{x}$, let $|\mathbf{x}|$ be the length of
$\mathbf{x}$. Introduce a partial order $\leqslant$ on $A^*$, let
$\mathbf{x}\leqslant\mathbf{y}$ if and only if there exists
$\mathbf{z}\in A^*$ such that $\mathbf{xz}=\mathbf{y}$.

Let $\mathcal{P}(A)$ be the set of subsets of $A$, including the
empty set $\emptyset$. Note that there is a
natural partial order on $\mathcal{P}(A)$, i.e., the subset order. Given a word $\mathbf{s}\in A^*$, let
$\mathbf{s}\omega$ be the set of letters of the word $\mathbf{s}$.
We say that $\mathbf{u},\mathbf{v}\in A^*$ are equivalent with
respect to $\omega$, $\mathbf{u}\sim_\omega\mathbf{v}$, if
$\mathbf{u}\omega=\mathbf{v}\omega$ (that is, $\mathbf{u}$ and
$\mathbf{v}$ consist of the same set of letters). Note that
$\sim_\omega$ is an equivalence relation.
The function $\omega$ is a morphism;
$(\mathbf{uv})\omega=\mathbf{u}\omega\cup\mathbf{v}\omega$.
Moreover, $\omega$ preserves the order relation since
$\mathbf{u}\leqslant\mathbf{v}$ implies
$\mathbf{u}\omega\subseteq\mathbf{v}\omega$.

Let $\mathcal{F}(A)$ be the set of all words over the alphabet $A$
that do not contain any repeated letters. The empty word
$\varepsilon$ is an element of $\mathcal{F}(A)$. Let $\tau$ be a
function such that for every $\mathbf{s}\in A^*$, any repeated
letters in $\mathbf{s}$ are deleted, leaving only the first
occurrence. We say that $\mathbf{u},\mathbf{v}\in A^*$ are
equivalent with respect to $\tau$, $\mathbf{u}\sim_\tau\mathbf{v}$,
if $\mathbf{u}\tau=\mathbf{v}\tau$. Note that $\sim_\tau$ is an
equivalence relation. Introduce a partial order $\leqslant$ on
$\mathcal{F}(A)$, let $\mathbf{v_1}\leqslant\mathbf{v_2}$ if and
only if there exists $\mathbf{v}\in\mathcal{F}(A)$ such that
$\mathbf{v_1}\mathbf{v}=\mathbf{v_2}$.

A deterministic finite automaton $\mathcal{A}$ is a tuple
$(Q,A,q_0,\cdp)$, where $Q$ - a set of states, $A$ - a finite
alphabet, $q_0$ - an initial state and $\cdp$ is a transition
function, that is, an everywhere defined function from $Q\times A$
to $Q$.  We say that a state $q$ of the
automaton $\mathcal{A}$ {\em accepts} a word $\mathbf{x}\in A^*$, if
the input $\mathbf{x}$ sets $\mathcal{A}$ into the state $q$. Given an automaton $(Q,A,q_0,\cdp)$, one may assign to it a
set of final states $Q_F$, a subset of $Q$. The resulting automaton
is denoted by $(Q,A,q_0,\cdp, Q_F)$.
\section{Monoids and Varieties}\label{variet}
A general overview on varieties of finite semigroups, monoids as
well as operations on them is given in \cite{P86}. It can also serve as a source for the definitions of morphisms and word quotients.

Unless specified otherwise, the monoids discussed in this section
are assumed to be finite.

An element $e$ of a monoid $\mathcal{M}$ is called an {\em idempotent}, if $e^2=e$.
It is a well-known fact that for any monoid $\mathcal{M}$ there exists $k>0$ such that for any element $x\in\mathcal{M}$ $x^k$ is idempotent. Moreover, if $x^k$ and $x^l$ both are idempotents, then $x^k=x^l$.
If $x$ is an element of a monoid $\mathcal{M}$, the unique
idempotent of the subsemigroup of $\mathcal{M}$ generated by $x$ is
denoted by $x^\omega$. The set of idempotents of the monoid
$\mathcal{M}$ is denoted by $E(\mathcal{M})$.

Given a regular language $L\subseteq A^*$, words $\mathbf{u},\mathbf{v}\in A^*$ are called {\em syntactically congruent}, $\mathbf{u}\sim_L\mathbf{v}$, if for all
$\mathbf{x},\mathbf{y}\in A^*$ $\mathbf{xuy}\in L$ if and only if $\mathbf{xvy}\in L$. The set of equivalence classes $A^*/\sim_L$ is a monoid, called {\em syntactic monoid} of $L$ and denoted $\mathcal{M}(L)$.
The morphism $\varphi$ from $A^*$ to $A^*/\sim_L$ is called {\em syntactic morphism}.

Given a monoid variety $\mathbf{V}$, the corresponding language
variety is denoted by $\boldsymbol{\mathcal{V}}$. The set of
languages over alphabet $A$ recognized by monoids in $\mathbf{V}$ is
denoted by $A^*\boldsymbol{\mathcal{V}}$.
\subsection{Varieties Definitions}\label{vardefs}
The monoid varieties used in this paper may be defined by some simple identities. For example, a monoid $\mathcal{M}$ belongs to the variety defined by an identity $[\![xy=yx]\!]$ if and only if for any
$x,y\in\mathcal{M}$ $xy=yx$.
In this paper, we shall refer to the following monoid varieties:

\begin{list}{(\arabic{Lcount})}{\usecounter{Lcount}\setlength{\itemindent}{\leftmargin}}
    \item $\mathbf{G} = [\![x^\omega=1]\!]$, the variety of groups.\\ The respective language variety is
    denoted $\boldsymbol{\mathcal{G}}$;

    \item $\mathbf{J_1} = [\![x^2=x,\ xy=yx]\!]$, the variety of
    commutative and idempotent monoids, also known as semilattice
    monoids.\\
    The respective language variety - $\boldsymbol{\mathcal{J}_1}$ (semilattice languages);

    \item $\mathbf{R_1}=[\![xyx=xy]\!]$, the variety of
    $\mathcal{R}$-trivial idempotent monoids, also known as left regular band
    monoids.
    The respective language variety - $\boldsymbol{\mathcal{R}_1}$ ($\mathcal{R}$-trivial
    idempotent languages, or $\mathcal{R}_1$ languages);

    \item $\mathbf{ER_1} = [\![x^\omega y^\omega x^\omega=x^\omega y^\omega]\!]$, the variety
    of such monoids $\mathcal{M}$ that $E(\mathcal{M})$ is an $\mathcal{R}$-trivial idempotent monoid.
    This
    variety is equal to $\mathbf{R_1*G}$ \cite{GP10}, the variety generated by semi\-direct products of $\mathcal{R}$-trivial
    idempotent monoids by groups. The respective language variety -
    $\boldsymbol{\mathcal{ER}_1}$;

    \item $\mathbf{J} = [\![x^\omega x=x^\omega,\
    (xy)^\omega=(yx)^\omega]\!] = [\![(xy)^\omega x=(xy)^\omega,\
    x(yx)^\omega=(yx)^\omega]\!]$, the variety of
    $\mathcal{J}$-trivial monoids. The respective language variety - $\boldsymbol{\mathcal{J}}$;

    \item $\mathbf{R}=[\![(xy)^\omega x=(xy)^\omega]\!]$, the variety of
    $\mathcal{R}$-trivial monoids.\\ The respective language variety -
    $\boldsymbol{\mathcal{R}}$;

    \item $\mathbf{EJ} = [\![(x^\omega y^\omega)^\omega=(y^\omega x^\omega)^\omega]\!]=
    [\![(x^\omega y^\omega)^\omega x^\omega=(x^\omega y^\omega)^\omega,\
    x^\omega(y^\omega x^\omega)^\omega=(y^\omega x^\omega)^\omega]\!]$, the variety of such monoids
    $\mathcal{M}$ that $E(\mathcal{M})$ generates a $\mathcal{J}$-trivial monoid. This
    variety is equal to $\mathbf{J*G}$, the variety generated by semi\-direct products of $\mathcal{J}$-trivial monoids
    by groups \cite{P95}.\\ The respective language variety - $\boldsymbol{\mathcal{EJ}}$;

    \item $\mathbf{ER} = [\![(x^\omega y^\omega)^\omega x^\omega=(x^\omega y^\omega)^\omega]\!]$, the variety considered
    in \cite{E76}. It is the variety
    of such monoids $\mathcal{M}$ that $E(\mathcal{M})$ generates an $\mathcal{R}$-trivial monoid \cite[p.132]{A94}. This
    variety is equal to $\mathbf{R*G}$, the variety generated by semi\-direct products of $\mathcal{R}$-trivial
    monoids by groups \cite[p.344]{A94}.\\ The respective language variety -
    $\boldsymbol{\mathcal{ER}}$.
\end{list}

It is possible to check that
$\mathbf{J_1}\subset\mathbf{J}\subset\mathbf{EJ}$,
$\mathbf{R_1}\subset\mathbf{R}\subset\mathbf{ER}$,
$\mathbf{R_1}\subset\mathbf{ER_1}\subset\mathbf{ER}$,
$\mathbf{J_1}\subset\mathbf{R_1}$, $\mathbf{J}\subset\mathbf{R}$ and
$\mathbf{G}\subset\mathbf{EJ}\subset\mathbf{ER}$.
\subsection{Semilattice Languages and Free Semilattices}
\label{J1subsec} We need some characterizations for semilattice
languages.

\begin{definition}
A free semilattice over an alphabet $A$ is a monoid
$(\mathcal{P}(A),\cup)$, where $\cup$ is the ordinary set union.
\end{definition}

For any alphabet $A$, the free semilattice $\mathcal{P}(A)$
satisfies the identities of $\mathbf{J_1}$, therefore
$\mathcal{P}(A)\in\mathbf{J_1}$.

For the sake of completeness, we give a proof for the following
\begin{proposition}
Given a language $L\in A^*\boldsymbol{\mathcal{J}_1}$, the free
semilattice $\mathcal{P}(A)$ is divided by $\mathcal{M}(L)$.
\end{proposition}
\begin{proof}
Let $\varphi$ be the syntactic morphism from $A^*$ to
$\mathcal{M}(L)$. It suffices to prove that $\omega^{-1}\varphi$ is
a surjective morphism.

Let $\mathbf{s_1},\mathbf{s_2}\in A^*$. Since
$\mathcal{M}(L)\in\mathbf{J_1}$,
$\mathbf{s_1}\sim_\omega\mathbf{s_2}$ implies
$\mathbf{s_1}\varphi=\mathbf{s_2}\varphi$. Let $p\in
\mathcal{P}(A)$. Let $\mathbf{t_1},\mathbf{t_2}\in p\omega^{-1}$.
Now, since $\mathbf{t_1}\sim_\omega\mathbf{t_2}$,
$\mathbf{t_1}\varphi=\mathbf{t_2}\varphi$. Hence
$\omega^{-1}\varphi$ is a function.

Let $p_1,p_2\in\mathcal{P}(A)$. Let $\mathbf{s_1}\in
(p_1p_2)\omega^{-1}$ and let $\mathbf{s_2}\in
(p_1\omega^{-1})(p_2\omega^{-1})$. The words $\mathbf{s_1}$ and
$\mathbf{s_2}$ consist of the same set of letters, so
$\mathbf{s_1}\sim_\omega\mathbf{s_2}$. Therefore
$\mathbf{s_1}\varphi=\mathbf{s_2}\varphi$, hence
$((p_1p_2)\omega^{-1})\varphi$ $=$
$((p_1\omega^{-1})(p_2\omega^{-1}))\varphi$ $=$
$p_1(\omega^{-1}\varphi)p_2(\omega^{-1}\varphi)$. So
$\omega^{-1}\varphi$ is a morphism.

The morphism $\varphi$ is surjective and $\omega$ is everywhere
defined, therefore $\omega^{-1}\varphi$ is surjective. \qed
\end{proof}
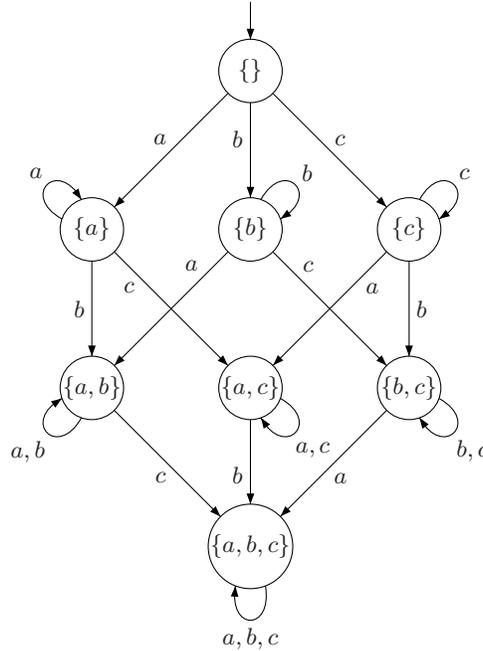
\begin{figure}[H]
    \begin{center}
        \unitlength=4pt
        \begin{picture}(30,54)(0,-6)
        \gasset{Nw=6,Nh=6,Nmr=6,curvedepth=0}
        \thinlines
        \node[Nmarks=i,iangle=90](E)(15,45){$\{\}$}
        \node(A)(0,30){$\{a\}$}
        \node(B)(15,30){$\{b\}$}
        \node(C)(30,30){$\{c\}$}
        \node(AB)(0,15){$\{a,b\}$}
        \node(AC)(15,15){$\{a,c\}$}
        \node(BC)(30,15){$\{b,c\}$}
        \gasset{Nw=8,Nh=8,Nmr=6,curvedepth=0}
        \node(ABC)(15,0){$\{a,b,c\}$}
        \drawloop[loopdiam=3,loopangle=135](A){$a$}
        \drawloop[loopdiam=3,loopangle=45](B){$b$}
        \drawloop[loopdiam=3,loopangle=45](C){$c$}
        \drawloop[loopdiam=3,loopangle=-135](AB){$a,b$}
        \drawloop[loopdiam=3,loopangle=-45](AC){$a,c$}
        \drawloop[loopdiam=3,loopangle=-45](BC){$b,c$}
        \drawloop[loopdiam=3,loopangle=-90](ABC){$a,b,c$}
        \drawedge[ELside=r,curvedepth=0](E,A){$a$}
        \drawedge[ELpos=43,ELside=r,curvedepth=0](E,B){$b$}
        \drawedge[ELside=l,curvedepth=0](E,C){$c$}
        \drawedge[ELside=r,curvedepth=0](A,AB){$b$}
        \drawedge[ELpos=30,ELside=r,curvedepth=0](A,AC){$c$}
        \drawedge[ELpos=30,ELside=r,curvedepth=0](B,AB){$a$}
        \drawedge[ELpos=30,ELside=l,curvedepth=0](B,BC){$c$}
        \drawedge[ELpos=30,ELside=l,curvedepth=0](C,AC){$a$}
        \drawedge[ELside=l,curvedepth=0](C,BC){$b$}
        \drawedge[ELside=r,curvedepth=0](AB,ABC){$c$}
        \drawedge[ELpos=55,ELside=r,curvedepth=0](AC,ABC){$b$}
        \drawedge[ELside=l,curvedepth=0](BC,ABC){$a$}
        \end{picture}
    \end{center}
    \caption{Free semilattice over $\{a,b,c\}$.}\label{fig.J1}
\end{figure}
An immediate consequence \cite[p.17, Prop. 2.7]{P86} is that
$\mathcal{P}(A)$ recognizes any language $L$ in
$A^*\boldsymbol{\mathcal{J}_1}$. Moreover, $L$ is a disjoint union
of some languages $X_1\omega^{-1}$,$...$, $X_n\omega^{-1}$, where
$X_1$,$...$, $X_n\in\mathcal{P}(A)$.

Thus, taking into account \cite[p.40, Prop. 3.10]{P86}, the following characterizations have been established:

\begin{theorem}\label{J1lang}
Let $L$ be a language over alphabet $A$. The following conditions are
equivalent:
\begin{conditions}
\item The syntactic monoid of $L$ belongs to the variety
$\mathbf{J_1}$;
\item $L$ is a Boolean combination of languages of the form
$A^*aA^*$, where $a\in A$;
\item $L$ is a Boolean combination of languages of the form
$B^*$, where $B\subseteq A$;
\item \label{J1disj} $L$ is a disjoint union of languages of the form
$X_1\omega^{-1},...,X_n\omega^{-1}$, where $X_1,...$,
$X_n\in\mathcal{P}(A)$.
\end{conditions}
\end{theorem}

Therefore, in order to specify a particular language $L\in
A^*\boldsymbol{\mathcal{J}_1}$, one may identify it by indicating a
particular subset of $\mathcal{P}(A)$.

Given a free semilattice $\mathcal{P}(A)$, one may represent it as a
deterministic finite automaton $(\mathcal{P}(A),A,\emptyset,\cdp)$,
where for every $X\in\mathcal{P}(A)$ and for every $a\in A$, $X\cdp
a=X\cup\{a\}$. By Theorem \ref{J1lang} (\ref{J1disj}), for any
semilattice language $L$ over alphabet $A$, $L\omega$ is a set of
final states, such that the automaton recognizes the language.

A free semilattice over $\{a,b,c\}$ represented as a finite
automaton is depicted in Figure \ref{fig.J1}.

The states of $(\mathcal{P}(A),A,\emptyset,\cdp)$ can be separated
into several levels, i.e., a state is at level $k$ if it corresponds
to an element in $\mathcal{P}(A)$ of cardinality $k$.
\subsection{$\mathcal{R}_1$ languages and Free Left Regular Bands}
\label{subsec_R1} We also need some characterizations for $\mathcal{R}_1$ languages.

\begin{definition}
A free left regular band over an alphabet $A$ is a monoid
$(\mathcal{F}(A),\cdp)$, where
$\mathbf{x}\cdp\mathbf{y}=(\mathbf{xy})\tau$, i.e., concatenation
followed by the application of $\tau$.
\end{definition}

The function $\tau$ is a morphism; for any $\mathbf{u},\mathbf{v}\in
A^*$ $(\mathbf{uv})\tau=\mathbf{u}\tau\cdp\mathbf{v}\tau$. Moreover,
$\tau$ preserves the order relation since
$\mathbf{u}\leqslant\mathbf{v}$ implies
$\mathbf{u}\tau\leqslant\mathbf{v}\tau$.

For any alphabet $A$, the free left regular band $\mathcal{F}(A)$
satisfies the identities of $\mathbf{R_1}$, therefore
$\mathcal{F}(A)\in\mathbf{R_1}$.

Characterizations of $\mathcal{R}_1$ languages are
established in \cite{PT84}:
\begin{theorem}\label{R1lang}
Let $L$ be a language over alphabet $A$. The following conditions are
equivalent:
\begin{conditions}
\item The syntactic monoid of $L$ belongs to the variety
$\mathbf{R_1}$;
\item $L$ is a Boolean combination of languages of the form
$B^*aA^*$, where $a\in A$ and $B\subseteq A$;
\item \label{R1disj} $L$ is a disjoint union of languages
of the form
$$a_1a_1^*a_2\{a_1,a_2\}^*a_3\{a_1,a_2,a_3\}^*...a_m\{a_1,a_2,...,a_m\}^*,$$
where the $a_i$'s are distinct letters of $A$.
\end{conditions}
\end{theorem}

Let $L$ be a single language from the disjoint union specified in
Theorem \ref{R1lang} (\ref{R1disj}). There exists a single
element $\mathbf{x}\in\mathcal{F}(A)$ such that
$\mathbf{x}\tau^{-1}=L$, therefore $\mathcal{F}(A)$ recognizes any
language in $A^*\boldsymbol{\mathcal{R}_1}$. Hence by \cite[p.17,
Prop. 2.7]{P86}, $\mathcal{M}(L)$ divides $\mathcal{F}(A)$.

Therefore, in order to specify a particular language $L\in
A^*\boldsymbol{\mathcal{R}_1}$, one may identify it by indicating a
particular subset of $\mathcal{F}(A)$. For example, the semilattice
language $A^*aA^*$ may also be denoted as $\bf{\{a,ab,ba,ac,ca,abc,acb,bac,bca,cab,}$ $\bf{cba\}}$.

It is also self-evident that $\mathcal{P}(A)$ is a quotient of
$\mathcal{F}(A)$. Indeed, let $\sigma$ be a restriction of $\omega$
to $\mathcal{F}(A)$. The function $\sigma$ is a surjective morphism from
$\mathcal{F}(A)$ to $\mathcal{P}(A)$ which preserves the order
relation.

Given a free left regular band $\mathcal{F}(A)$, one may represent
it as a deterministic finite automaton
$(\mathcal{F}(A),A,\varepsilon,\cdp_{\mathcal{F}(A)})$. By Theorem
\ref{R1lang} (\ref{R1disj}), for any $\mathcal{R}_1$ language $L$ over alphabet $A$, $L\tau$ is a set of final
states, such that the automaton recognizes the language.

A free left regular band over $\{a,b,c\}$ represented as a finite
automaton is depicted in Figure \ref{fig.R1}.

Free left regular bands and free semilattices are key elements to prove that a quantum automaton may recognize a particular
$\mathcal{R}_1$ language if and only if its system of linear inequalities is consistent.
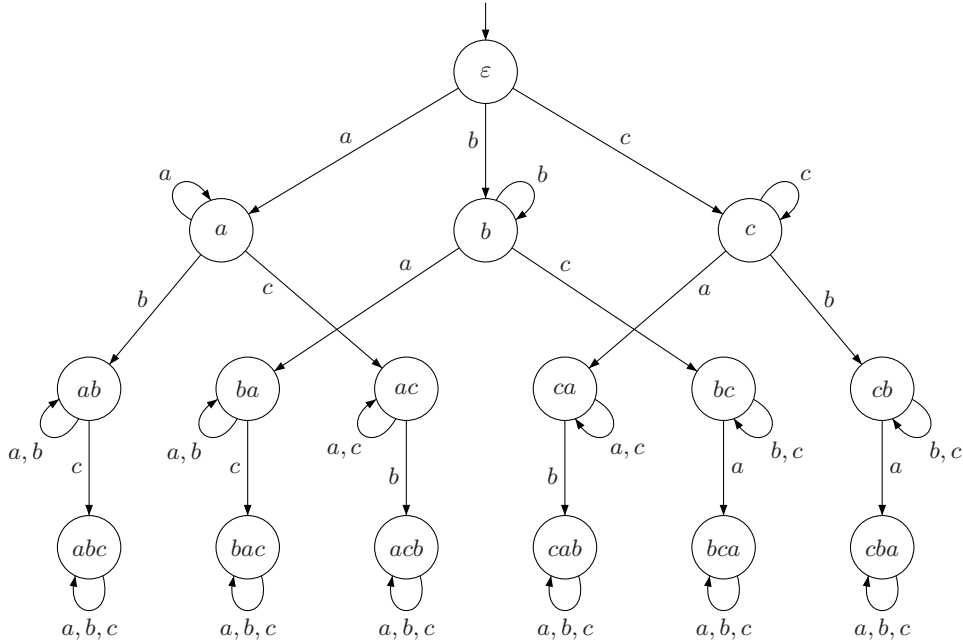
\begin{figure}[H]
    \begin{center}
        \unitlength=4pt
        \begin{picture}(70,54)(2,-4)
        \gasset{Nw=6,Nh=6,Nmr=6,curvedepth=0}
        \thinlines
        \node[Nmarks=i,iangle=90](E)(37.5,45){$\varepsilon$}
        \node(A)(12.5,30){$a$}
        \node(B)(37.5,30){$b$}
        \node(C)(62.5,30){$c$}
        \node(AB)(0,15){$ab$}
        \node(BA)(15,15){$ba$}
        \node(AC)(30,15){$ac$}
        \node(CA)(45,15){$ca$}
        \node(BC)(60,15){$bc$}
        \node(CB)(75,15){$cb$}
        \node(ABC)(0,0){$abc$}
        \node(BAC)(15,0){$bac$}
        \node(ACB)(30,0){$acb$}
        \node(CAB)(45,0){$cab$}
        \node(BCA)(60,0){$bca$}
        \node(CBA)(75,0){$cba$}
        \drawloop[loopdiam=3,loopangle=135](A){$a$}
        \drawloop[loopdiam=3,loopangle=45](B){$b$}
        \drawloop[loopdiam=3,loopangle=45](C){$c$}
        \drawloop[loopdiam=3,loopangle=-135](AB){$a,b$}
        \drawloop[loopdiam=3,loopangle=-135](BA){$a,b$}
        \drawloop[loopdiam=3,loopangle=-135](AC){$a,c$}
        \drawloop[loopdiam=3,loopangle=-45](CA){$a,c$}
        \drawloop[loopdiam=3,loopangle=-45](BC){$b,c$}
        \drawloop[loopdiam=3,loopangle=-45](CB){$b,c$}
        \drawloop[loopdiam=3,loopangle=-90](ABC){$a,b,c$}
        \drawloop[loopdiam=3,loopangle=-90](BAC){$a,b,c$}
        \drawloop[loopdiam=3,loopangle=-90](ACB){$a,b,c$}
        \drawloop[loopdiam=3,loopangle=-90](CAB){$a,b,c$}
        \drawloop[loopdiam=3,loopangle=-90](BCA){$a,b,c$}
        \drawloop[loopdiam=3,loopangle=-90](CBA){$a,b,c$}
        \drawedge[ELside=r,curvedepth=0](E,A){$a$}
        \drawedge[ELpos=43,ELside=r,curvedepth=0](E,B){$b$}
        \drawedge[ELside=l,curvedepth=0](E,C){$c$}
        \drawedge[ELside=r,curvedepth=0](A,AB){$b$}
        \drawedge[ELpos=30,ELside=r,curvedepth=0](A,AC){$c$}
        \drawedge[ELpos=30,ELside=r,curvedepth=0](B,BA){$a$}
        \drawedge[ELpos=30,ELside=l,curvedepth=0](B,BC){$c$}
        \drawedge[ELpos=30,ELside=l,curvedepth=0](C,CA){$a$}
        \drawedge[ELside=l,curvedepth=0](C,CB){$b$}
        \drawedge[ELside=r,curvedepth=0](AB,ABC){$c$}
        \drawedge[ELside=r,curvedepth=0](BA,BAC){$c$}
        \drawedge[ELpos=55,ELside=r,curvedepth=0](AC,ACB){$b$}
        \drawedge[ELpos=55,ELside=r,curvedepth=0](CA,CAB){$b$}
        \drawedge[ELside=l,curvedepth=0](BC,BCA){$a$}
        \drawedge[ELside=l,curvedepth=0](CB,CBA){$a$}
        \end{picture}
    \end{center}
    \caption{Free left regular band over $\{a,b,c\}$.}\label{fig.R1}
\end{figure}
\section{Completely Positive Maps} \label{sec_cp_maps}
In this section, we establish some facts about completely positive
maps with certain properties, i.e., completely positive maps that
describe the evolution of bistochastic quantum finite automata,
defined in the next section. A comprehensive account on quantum
computation can be found in \cite{NC00}.

Following \cite{NC00}, we call a matrix $M\in\mathbb{C}^{n\times n}$ {\em positive}, if
for any vector $X\in\mathbb{C}^n$, $X^*MX$ is real and nonnegative.
In literature, positive matrices sometimes are called positive
semi-definite. For arbitrary matrices $M,N$ we may write $M\geqslant
N$ if $M-N$ is positive. This defines a partial ordering on
$\mathbb{C}^{n\times n}$. Also note that the set of all positive
matrices in $\mathbb{C}^{n\times n}$ is an additive monoid. A matrix
is positive if and only if it is Hermitian and all of its
eigenvalues are nonnegative (\cite[Exercises 2.21,2.24]{NC00}). A
matrix $M$ is positive if and only if exists a matrix
$S\in\mathbb{C}^{n\times n}$ such that $M=S^*S$ (\cite[Section
6.1]{Z99}). Let $\operatorname{Tr}(A)$ be a trace of a matrix $A$.
The inner product of two matrices $A$ and $B$ is defined as $\langle
A,B\rangle=\operatorname{Tr}(A^*B)$. Consequently, the norm of a
matrix $A$ (the Frobenius norm) is defined as
$\|A\|=\sqrt{\operatorname{Tr}(A^*A)}$.
\begin{proposition}\label{posit_crit}
A matrix $M$ is positive if and only if for any positive $A$
$\operatorname{Tr}(MA)\geqslant0$.
\end{proposition}
\begin{proof}
Assume that for any positive $A$ $\operatorname{Tr}(MA)\geqslant0$.
Take $A=XX^*$, where $X$ is an arbitrary vector. Now
$X^*MX=\operatorname{Tr}(MXX^*)\geqslant0$, hence for any $X$
$X^*MX\geqslant0$. So $M$ is positive.

Assume $M$ is positive. Let $A$ be a positive matrix. So $A$ admits
spectral decomposition, $A=\sum\limits_{i=1}^n
\lambda_iX_iX_i^*$, where $\lambda_i$ are nonnegative eigenvalues and $X_i$ -
eigenvectors of $A$. Now
$\operatorname{Tr}(MA)=\operatorname{Tr}(M\sum\limits_{i=1}^n
\lambda_iX_iX_i^*)=\sum\limits_{i=1}^n\lambda_i\operatorname{Tr}(MX_iX_i^*)=\sum\limits_{i=1}^n\lambda_iX_i^*MX_i\geqslant0$.
\qed
\end{proof}

A linear map $\Phi:\mathbb{C}^{n\times
n}\longrightarrow\mathbb{C}^{m\times m}$ is called {\em positive},
if for any $n\times n$ positive matrix $M$ $\Phi(M)$ is positive.
Any linear map from $\mathbb{C}^{n\times n}$ to $\mathbb{C}^{m\times
m}$ may be regarded as a linear operator in $\mathbb{C}^{n^2\times
m^2}$. The norm of a linear map $\Phi$ from $\mathbb{C}^{n\times n}$
to $\mathbb{C}^{m\times m}$ is defined as
$\sup\limits_{M\in\mathbb{C}^{n\times n}}\frac{\|\Phi(M)\|}{\|M\|}$.
A linear map $\Phi$ is called a {\em contraction}, if
$\|\Phi\|\leqslant1$.

Let $I_s$ be the identity map over $\mathbb{C}^{s\times s}$. Given
two linear maps $\Phi$ and $\Psi$, let $\Phi\bigotimes\Psi$ be the
tensor product of those maps. A positive linear map $\Phi$ is called
{\em completely positive (CP)}, if for any $s\geqslant1$,
$\Phi\bigotimes I_s$ is positive. By Choi's theorem \cite{C75}, a
linear map is completely positive if and only if it admits a {\em
Kraus decomposition}, meaning that there exist matrices
$V_1,\dots,V_l\in\mathbb{C}^{m\times n}$, $l\leqslant nm$, such that
for any matrix $M\in\mathbb{C}^{n\times n}$
$\Phi(M)=\sum\limits_{i=1}^l V_iMV_i^*$. So any CP map may be
identified by a set of its Kraus operators $\{V_1,\dots,V_l\}$.

A completely positive map $\Phi$ is called {\em trace-preserving},
if for any positive $M$,
$\operatorname{Tr}(\Phi(M))=\operatorname{Tr}(M)$. A CP map
$\Phi=\{V_1,\dots,V_l\}$ from $\mathbb{C}^{n\times n}$ to
$\mathbb{C}^{m\times m}$ is trace preserving if and only if
$\sum\limits_{i=1}^l V_i^*V_i=I_n$ \cite[\textsection8.2.3]{NC00}.

A completely positive map $\Phi$ is called {\em sub-tracial} iff for
any positive $M$ we have
$\operatorname{Tr}(\Phi(M))\leqslant\operatorname{Tr}(M)$.
\begin{theorem}
A completely positive map $\Phi=\{V_1,\dots,V_l\}$ from
$\mathbb{C}^{n\times n}$ to $\mathbb{C}^{m\times m}$ is sub-tracial
if and only if $\sum\limits_{i=1}^l V_i^*V_i\leqslant I_n$.
\end{theorem}
\begin{proof}
Assume that $\sum\limits_{i=1}^l V_i^*V_i\leqslant I_n$. So exists a
positive matrix $P$ such that $\sum\limits_{i=1}^l V_i^*V_i+P=I_n$.
Moreover, $P=\sum\limits_{i=1}^n
\lambda_iX_iX_i^*$, where $\lambda_i$ are nonnegative eigenvalues
and $X_i$ - eigenvectors of $P$. By adding $m-1$ zero columns to each
vector $X_i$ one respectively obtains matrices
$W_i\in\mathbb{C}^{n\times m}$ such that $X_iX_i^*=W_iW_i^*$. For
each $i$, $1\leqslant i\leqslant n$, let
$V_{l+i}=\sqrt{\lambda_i}W_i^*$. So $\sum\limits_{i=1}^{l+n}
V_i^*V_i=I_n$. Hence $\{V_1,\dots,V_l,\dots,V_{l+n}\}$ is a
trace-preserving CP map, so for any positive $M$
$\operatorname{Tr}\left(\sum\limits_{i=1}^{l+n}V_iMV_i^*\right)=\operatorname{Tr}(M)$.
The matrix $\sum\limits_{i=1}^nV_{l+i}MV_{l+i}^*$ is positive,
therefore
$\operatorname{Tr}\left(\sum\limits_{i=1}^nV_{l+i}MV_{l+i}^*\right)\geqslant0$.
Hence
$\operatorname{Tr}\left(\sum\limits_{i=1}^lV_iMV_i^*\right)\leqslant\operatorname{Tr}(M)$.

Assume that for all positive $M$
$\operatorname{Tr}(\Phi(M))\leqslant\operatorname{Tr}(M)$. Since
$\operatorname{Tr}\left(\sum\limits_{i=1}^l
V_iMV_i^*\right)=\operatorname{Tr}\left(\sum\limits_{i=1}^l
V_i^*V_iM\right)$, for all positive $M$
$\operatorname{Tr}\left(\sum\limits_{i=1}^l
V_i^*V_iM\right)\leqslant\operatorname{Tr}\left(M\right)$. So for
any positive $M$
$\operatorname{Tr}\left(\left(I_n-\sum\limits_{i=1}^l
V_i^*V_i\right)M\right)\geqslant0$. Now by Proposition \ref{posit_crit},
$I_n-\sum\limits_{i=1}^l V_i^*V_i$ is positive, therefore
$\sum\limits_{i=1}^l V_i^*V_i\leqslant I_n$. \qed
\end{proof}
A CP map $\Phi=\{V_1,\dots,V_l\}$ from $\mathbb{C}^{n\times n}$ to
$\mathbb{C}^{m\times m}$ is called {\em unital} if $\Phi(I_n)=I_m$,
i.e., $\sum\limits_{i=1}^l V_iV_i^*=I_m$. A CP map from
$\mathbb{C}^{n\times n}$ to $\mathbb{C}^{m\times m}$
$\Phi=\{V_1,\dots,V_l\}$ is called {\em sub-unital} if
$\Phi(I_n)\leqslant I_m$, i.e., $\sum\limits_{i=1}^l
V_iV_i^*\leqslant I_m$.

A {\em composition} of CP maps $\Phi_0,...,\Phi_m$ from $\mathbb{C}^{n\times n}$ to $\mathbb{C}^{n\times n}$ is a CP map $\Phi=\Phi_0\circ\dots\circ\Phi_m$ such that for any $M\in\mathbb{C}^{n\times n}$
$\Phi(M)=\Phi_0(\Phi_1(...(\Phi_m(M))...)$.

A CP map $\Phi=\{V_1,\dots,V_l\}$ from $\mathbb{C}^{n\times n}$ to
$\mathbb{C}^{n\times n}$ is called {\em bistochastic}, if it is both
trace preserving and unital, i.e.,  $\sum\limits_{i=1}^l
V_iV_i^*=\sum\limits_{i=1}^l V_i^*V_i=I_n$.

{\em Examples of bistochastic CP maps}.
\begin{list}{(\arabic{Lcount})}{\usecounter{Lcount}\setlength{\itemindent}{\leftmargin}}
\item A map defined by unitary matrix $U$, i.e., a CP map
 $\Phi(M)=UMU^*$, called {\em unitary operation};
\item A collection of projection matrices $\{P_i\}$ such that
$\sum\limits_{i=1}^lP_i=I$, i.e., a CP map
$\Phi(M)=\sum\limits_{i=1}^lP_iMP_i^*$, called {\em orthogonal
measurement};
\item A CP map $\Phi(M)=\sum\limits_{i=1}^lp_iU_iMU_i^*$, where
$\sum\limits_{i=1}^lp_i=1$ and for all $i$ $U_i $ are unitary. Such
a map is called {\em random unitary operation};
\item Any composition of the maps above.
\end{list}

A CP map $\mathbb{C}^{n\times n}$ to $\mathbb{C}^{n\times n}$ is
called {\em sub-bistochastic}\footnote{Sometimes in quantum physics and quantum computation literature, a CP map is sub-tracial by definition. In such cases, sub-bistochastic CP maps are called sub-unital
CP.}, if it is both sub-unital and
sub-tracial. A composition of two sub-bistochastic CP maps
is a sub-bistochastic CP map.

We are interested about
some properties of the asymptotic dynamics resulting from iterative
application of a CP sub-bistochastic map.

A CP map $\Phi$ from $\mathbb{C}^{n\times n}$ to
$\mathbb{C}^{n\times n}$ is called {\em idempotent} if
$\Phi\circ\Phi=\Phi$.

\begin{definition}
A CP map $\Phi$ from $\mathbb{C}^{n\times n}$ to
$\mathbb{C}^{n\times n}$ generates a unique idempotent, denoted
$\Phi^\omega$, if there exists a sequence of positive integers $n_s$
such that 1) exists the limit
$\Phi^\omega=\lim\limits_{s\to\infty}\Phi^{n_s}$; 2) the CP map
$\Phi^\omega$ is idempotent; 3) for any sequence of positive
integers $m_s$ such that the limit
$\lim\limits_{s\to\infty}\Phi^{m_s}$ exists and is idempotent,
$\lim\limits_{s\to\infty}\Phi^{m_s}=\Phi^\omega$.
\end{definition}

For example, if $\Phi$ is a unitary operation then $\Phi^\omega$ is
the identity map. (Theorem \ref{theor_idemp_exists}.)

Note that any CP map from $\mathbb{C}^{n\times n}$ to
$\mathbb{C}^{n\times n}$ may be regarded as a linear operator in
$\mathbb{C}^{n^2\times n^2}$. In this sense, the conjugate transpose of $\Phi=\{V_1,\dots,V_l\}$ is $\Phi^*=\{V_1^*,\dots,V_l^*\}$.
Kuperberg has provided a sketch of the proof \cite{K03} that for any CP
sub-bistochastic map $\Phi$ from $\mathbb{C}^{n\times n}$ to
$\mathbb{C}^{n\times n}$, its idempotent $\Phi^\omega$ exists and it
is a linear projection operator in $\mathbb{C}^{n^2\times n^2}$. We reconstruct a full proof of this result below.  The
first step in that direction is the following theorem.
\begin{theorem}\label{theor_contract}
Any CP sub-bistochastic map $\Phi$ is a contraction.\footnote{The special case dealing with bistochastic maps was proved in \cite{PG06}.}
\end{theorem}
\begin{proof}
We need to prove that $\|\Phi\|\leqslant1$. Let $\sigma_{\max}(\Phi)$ - the largest singular value of $\Phi$ and
$\lambda_{\max}(\Phi^*{\circ}\,\Phi)$ - the largest eigenvalue of $\Phi^*{\circ}\,\Phi$.
Note that $\|\Phi\|=\sigma_{\max}(\Phi)=\sqrt{\lambda_{\max}(\Phi^*{\circ}\,\Phi)}$.
Let $M$ an eigenvector of $\Phi^*{\circ}\,\Phi$ corresponding to $\lambda_{\max}$. So $\Phi^*{\circ}\,\Phi(M)=\lambda_{\max}M$.
Suppose $M$ is not Hermitian.
Let $V_1,...,V_l$ be the Kraus operators corresponding to $\Phi^*{\circ}\,\Phi$. So $\Phi^*{\circ}\,\Phi(M^*)=\sum\limits_{i=1}^l V_iM^*V_i^*=(\sum\limits_{i=1}^l V_iMV_i^*)^*=(\lambda_{\max}M)^*=\lambda_{\max}M^*$.
Hence $M^*$ is an eigenvector corresponding to $\lambda_{\max}$ as well. Therefore $M+M^*$ is an eigenvector also corresponding to $\lambda_{\max}$, and it is Hermitian.
So without loss of generality, we may assume that $M$ is Hermitian.  Note that $\operatorname{Tr}(\Phi^*{\circ}\,\Phi(M))=\lambda_{\max}\operatorname{Tr}(M)$. On the other hand, since $\Phi^*{\circ}\,\Phi$ is sub-bistochastic,
$\operatorname{Tr}(\Phi^*{\circ}\,\Phi(M))\leqslant\operatorname{Tr}(M)$. Hence $\lambda_{\max}\leqslant1$. Therefore $\|\Phi\|\leqslant1$.
\qed
\end{proof}
\begin{theorem}
\label{theor_idemp_exists}
Any CP sub-bistochastic map $\Phi$ generates a unique idempotent $\Phi^\omega$.
\end{theorem}
\begin{proof}
Let $\sigma_{\max}$ - the largest singular value of $\Phi$ and $\lambda$ - any of its eigenvalues.
Due to Browne's theorem \cite[Fact 5.11.21 i)]{B09}, $|\lambda|\leqslant\sigma_{\max}$. Therefore by Theorem \ref{theor_contract}, $|\lambda|\leqslant1$.
Let $\lambda_1$ an eigenvalue such that $|\lambda_1|=1$. Let $m(\lambda_1)$ and $g(\lambda_1)$ - the algebraic and geometric multiplicity of $\lambda_1$. It has been proved in \cite[Lemmas 2 and 3]{LP10}
that $m(\lambda_1)=g(\lambda_1)$.
(The proofs are given for the bistochastic case, but they can be copied for sub-bistochastic case with a sole modification: in the proof of Lemma 2 in \cite{LP10}, replace "$\|\mathbf{\Phi}_\mathcal{A}\|=1$"
with "$\|\mathbf{\Phi}_\mathcal{A}\|\leqslant1$". Before \cite{LP10}, the same has been proved for random unitary operations in \cite{NAJ10}.)

The map $\Phi$ may be viewed as an $n^2\times n^2$ matrix, it admits Jordan normal form. So $\Phi=SJS^{-1}$, where $J$ is a Jordan block matrix and $S$ - some non-singular matrix.
Consider the Jordan blocks corresponding to any eigenvalue $\lambda_1$ such that $|\lambda_1|=1$. Since $m(\lambda_1)=g(\lambda_1)$, any such Jordan block is one-dimensional.
Any other Jordan block $B$ is related to an eigenvalue $\lambda$ such that $|\lambda|<1$, so $\lim\limits_{s\to\infty}B^s=0$.
Consider the diagonal matrix $L$ corresponding to eigenvalues $\lambda_1$ such that $|\lambda_1|=1$. There exists a strictly monotone increasing sequence of positive integers $n_s$ such that
$\lim\limits_{s\to\infty}L^{n_s}=I$. (This is implied by \cite[Theorem 201]{HW79}.) Thus $L^\omega=I$. The uniqueness of $L^\omega$ comes from the fact that the identity matrix is the only idempotent diagonal matrix with diagonal entries all nonzero.
So $J$ generates a unique idempotent;
$J^\omega$ is a diagonal matrix with zeroes and ones on the diagonal.
Therefore $\Phi$ generates a unique idempotent as well; $\Phi^\omega=SJ^\omega S^{-1}$.
\qed
\end{proof}
If $\Phi$ is a CP  sub-bistochastic map, then $\Phi^\omega$ is a CP sub-bistochastic map as well.
\begin{theorem}
The unique idempotent $\Phi^\omega$ generated by a CP sub-bistochastic map $\Phi$ from $\mathbb{C}^{n\times n}$ to
$\mathbb{C}^{n\times n}$ is a projection operator in $\mathbb{C}^{n^2\times n^2}$.
\end{theorem}
\begin{proof}
By Theorem \ref{theor_contract}, $\Phi$ is a contraction. So $\Phi^\omega$ is a contraction as well. Therefore, due to Halperin \cite[3.(III)]{H62}, $\Phi^\omega$ is a projection.
\qed
\end{proof}
Finally, we are ready to formulate a theorem, which is the main result of this section. As shown further in the paper, this theorem ultimately is the reason why certain models of
quantum finite automata cannot recognize all regular languages.
\begin{theorem}
\label{theor_bist_EJ}
Let $e_1,...,e_k$ be idempotent CP sub-bistochastic maps from $\mathbb{C}^{n\times n}$ to
$\mathbb{C}^{n\times n}$. Then for any
$i$, $1\leqslant i\leqslant k$,
\begin{conditions}
\item $\lim\limits_{n\to\infty}(e_1\circ...\circ e_k)^n=(e_1\circ...\circ e_k)^\omega=(e_{\pi(1)}\circ...\circ e_{\pi(k)})^\omega$, where $\pi$ is a permutation in $\{1,\dots,k\}$;
\item $(e_1\circ...\circ e_k)^\omega=e_i\circ(e_1\circ...\circ e_k)^\omega=(e_1\circ...\circ e_k)^\omega\circ e_i$.
\end{conditions}
\end{theorem}
\begin{proof}
Since $e_1,...,e_k$ are projections, by von Neumann-Halperin theorem \cite[Theorem 1]{H62}, $\lim\limits_{n\to\infty}(e_1\circ...\circ e_k)^n=(e_1\circ...\circ e_k)^\omega=(e_{\pi(1)}\circ...\circ e_{\pi(k)})^\omega$. In the same way,
$(e_1\circ...\circ e_k)^\omega=(e_i\circ e_1\circ...\circ e_k)^\omega=(e_1\circ...\circ e_k\circ e_i)^\omega.$
Note that $e_i\circ(e_i\circ e_1\circ...\circ e_k)^\omega=(e_i\circ e_1\circ...\circ e_k)^\omega$ and $(e_1\circ...\circ e_k\circ e_i)^\omega\circ e_i=(e_1\circ...\circ e_k\circ e_i)^\omega.$
Therefore $(e_1\circ...\circ e_k)^\omega=e_i\circ(e_1\circ...\circ e_k)^\omega=(e_1\circ...\circ e_k)^\omega\circ e_i.$
\qed
\end{proof}

Any finite quantum system at a particular moment of time (i.e., its {\em mixed state}) is
described by a {\em density matrix}. By \cite[Theorem 2.5]{NC00}, a
matrix is a density matrix if and only if it is positive and its
trace is equal to $1$.

Informally, an $n\times n$ density matrix describes a quantum system
with $n$ states. A completely positive trace-preserving map describes an evolution of
a quantum system as allowed by quantum mechanics. It maps a density
matrix to a density matrix.
\section{Automata Models}\label{sec_automata}
An overview of different models of finite automata, relevant to our
research, is given in the following table. The definition for bistochastic quantum finite automata is given below. For the formal
definitions of other indicated automata models, the reader is referred
to the references given in the table.

As seen further, measure-once (measure-many) bistochastic quantum finite automata is a generalization of any other "classical" ("decide-and-halt", respectively) word acceptance model from Table \ref{table_automata}.
At the same time BQFA have the same limitations for language recognition as known for other models above.
Thus we consider the introduction of yet another model of quantum finite automata justified, because it allows us to prove the limitations of language recognition for all the models within single framework.
Therefore the proof of the new limitations for MM-BQFA in Section \ref{Lin_Ineq}, which are expressed in terms of linear inequalities, implies the same for any other "decide-and-halt" word acceptance model in the table.
\hbadness=20000 \vbadness=20000
\begin{table}
\caption{Automata Models}
\label{table_automata}
\begin{tabular}{|p{3.2cm}|p{4.2cm}|p{4.7cm}|}
\hline
 & \mbox{"Classical" word}\newline \mbox{acceptance}
 & \mbox{"Decide-and-halt" word}\newline \mbox{acceptance}
\\
\hline
Deterministic \mbox{Reversible Automata} & \mbox{Group
Automata (GA)} \mbox{\cite{HS66,T68}} &
\mbox{Reversible Finite Automata} (RFA) \cite{AF98,GP10}\\
\hline
\mbox{Quantum Finite} \mbox{Automata with Pure} States &
\mbox{Measure-Once Quantum} \mbox{Finite Automata} \mbox{(MO-QFA)}
\cite{MC97,BP99} &
\mbox{Measure-Many Quantum} \mbox{Finite Automata (MM-QFA)} \cite{KW97,BP99,AV00,AT04}\\
\hline
Probabilistic \mbox{Reversible Automata} & \mbox{"Classical"
Probabilistic} \mbox{Reversible Automata} \mbox{(C-PRA)}
\cite{GK02,AT04}&
\mbox{"Decide-and-halt" Probabilistic} \mbox{Reversible Automata (DH-PRA)} \cite{GK02,GK09}\\
\hline
\multirow{2}{*}{\newline\begin{minipage}[b]{2in}\ \\Quantum Finite \\ Automata with \\Mixed States\end{minipage}} &
\mbox{Latvian Quantum Finite} \mbox{Automata} (LQFA) \cite{AT04}&
\mbox{Enhanced Quantum Finite} \mbox{Automata} \mbox{(EQFA)} \cite{N99,M08}\\
\cline{2-3}
 & \mbox{Measure-Once Bistochastic} \mbox{Quantum Finite Automata} (MO-BQFA)&
\mbox{Measure-Many Bistochastic} \mbox{Quantum Finite Automata} \mbox{(MM-BQFA)}\\
\hline
\end{tabular}
\end{table}
\hbadness=1000 \vbadness=1000
\begin{definition}
\label{def_BQFA}
A bistochastic quantum finite automaton (BQFA) is a tuple $(Q,A\cup\{\#,\$\},q_0,\{\Phi_a\})$, where $Q$ is a finite set of states, $A$ - a finite input alphabet, $\#,\$\notin A$ - initial and final end-markers,
$q_0$ - an initial state and for each $a\in A\cup\{\#,\$\}$ $\Phi_a$ is
a CP bistochastic transition map from $\mathbb{C}^{|Q|\times|Q|}$ to $\mathbb{C}^{|Q|\times|Q|}$.
\end{definition}
Regardless of which word acceptance model is used, each input word is enclosed into end-markers $\#,\$$. At any step, the mixed state of a BQFA may be described by a density matrix $\rho$.
The computation starts in the state $|q_0\rangle\langle q_0|$.

{\em Operation of a measure-once BQFA and word acceptance.}
On input letter $a\in A$, $\rho$ is transformed into $\Phi_a(\rho)$.
The set of states $Q$ is partitioned into two disjoint subsets $Q_{acc}$ and $Q_{rej}$. After reading the final end-marker $\$$, a measurement $\{P_{acc},P_{rej}\}$ is applied to $\rho$,
where $P_{acc}=\sum\limits_{q\in{Q_{acc}}}|q\rangle\langle q|$ and $P_{rej}=\sum\limits_{q\in{Q_{rej}}}|q\rangle\langle q|$. The respective input word is accepted with probability $\operatorname{Tr}(P_{acc}\rho P_{acc})$
and rejected with probability $\operatorname{Tr}(P_{rej}\rho P_{rej})$. For any word $\mathbf{a}=a_1\dots a_k$, define $\Phi_\mathbf{a}=\Phi_{a_k}\circ\dots\circ\Phi_{a_1}$.

{\em Operation of a measure-many BQFA and word acceptance.}
The set of states $Q$ is partitioned into three disjoint subsets $Q_{non}$, $Q_{acc}$ and $Q_{rej}$ - non-halting, accepting and rejecting states, respectively. It is assumed that $q_0\in Q_{non}$.
On input letter $a\in A$, $\rho$ is transformed into $\rho^\prime=\Phi_a(\rho)$. After that, a measurement $\{P_{non},P_{acc},P_{rej}\}$ is applied to $\rho^\prime$,
where for each $i\in\{non,acc,rej\}$ $P_i=\sum\limits_{q\in{Q_i}}|q\rangle\langle q|$.
The respective input word is
accepted (rejected) with probability $\operatorname{Tr}(P_{acc}\rho^\prime P_{acc})$ ($\operatorname{Tr}(P_{rej}\rho^\prime P_{rej})$, respectively). If the input word is accepted or rejected, the computation is halted.
Otherwise, with probability $\operatorname{Tr}(P_{non}\rho^\prime P_{non})$, the computation continues
from the mixed state $P_{non}\rho^\prime P_{non}/\operatorname{Tr}(P_{non}\rho^\prime P_{non})$. To ensure that any input word is always either accepted or rejected, it is required for $\Phi_\$$ that for any $\rho$
such that $\operatorname{Tr}(P_{non}\rho P_{non})=1$, $\operatorname{Tr}(P_{non}\Phi_\$(\rho)P_{non})=0$.

To describe the probability distribution $S_{\#\mathbf{u}}$ of a MM-BQFA $\mathcal{A}$ after reading some prefix $\#\mathbf{u}$, it is convenient to use density matrices $\rho$ scaled by $p$, $0\leqslant p\leqslant1$.
So the probability distribution $S_{\#\mathbf{u}}$ of $\mathcal{A}$ is a triple $(\rho,p_{acc},p_{rej})$, where $\operatorname{Tr}(\rho)+p_{acc}+p_{rej}=1$, $\rho/\operatorname{Tr}(\rho)$ is the current mixed state and
$p_{acc},p_{rej}$ are respectively the probabilities that $\mathcal{A}$ has accepted or rejected the input. So the scaled density matrix $\rho$ may be called a {\em scaled mixed state}.
For any $a\in A\cup\{\#,\$\}$, let $\Psi_a(\rho)=P_{non}\Phi_a(\rho)P_{non}$.
After reading the next input letter $a$, the probability distribution is
$S_{\#\mathbf{u}a}=(\Psi_a(\rho),p_{acc}+\operatorname{Tr}(P_{acc}\Phi_a(\rho)P_{acc}),p_{rej}+\operatorname{Tr}(P_{rej}\Phi_a(\rho)P_{rej}))$.
For any word $\mathbf{a}=a_1\dots a_k$, define $\Psi_\mathbf{a}=\Psi_{a_k}\circ\dots\circ\Psi_{a_1}$. Hence $\rho=\Psi_{\#\mathbf{u}}(|q_0\rangle\langle q_0|)$. Note that $\Psi_\mathbf{a}$ is a CP sub-bistochastic map.
%
%

{\em Language recognition} is defined in a way equivalent to Rabin's \cite{R63}. Suppose that an automaton $\mathcal{A}$ corresponds to one of the
probabilistic or quantum models from the table above.
By $p_{\mathbf{x},\mathcal{A}}$ (or $p_\mathbf{x}$, if no ambiguity arises) we denote the probability that an
input $\mathbf{x}$ is accepted by the automaton $\mathcal{A}$.
Furthermore, we denote $P_L=\{p_{\mathbf{x},\mathcal{A}}\ |\
\mathbf{x}\in L\}$, $\overline{P_L}=\{p_{\mathbf{x},\mathcal{A}}\ |\
\mathbf{x}\notin L\}$, $p_1=\sup\overline{P_L}$, $p_2=\inf P_L$.
It is said that an automaton $\mathcal{A}$ recognizes a language $L$
with interval $(p_1,p_2)$, if $p_1\leq p_2$ and
$P_L\cap\overline{P_L}=\emptyset$.
It is said that an automaton $\mathcal{A}$ recognizes a language $L$
with bounded error and interval $(p_1,p_2)$, if $p_1<p_2$.
We consider only bounded error language recognition.
An automaton is said to recognize a language with probability $p$ if
the automaton recognizes the language with interval $(1-p,p)$.
It is said that a language is recognized by some class of automata with
probability $1-\epsilon$, if for every $\epsilon>0$ there
exists an automaton in the class which recognizes the language with
interval $(\epsilon_1,1-\epsilon_2)$, where
$\epsilon_1,\epsilon_2\leq\epsilon$. A language $L$ is recognizable with interval $(p_1,p_2)$ iff it is
recognizable with some probability $p$ (see, for example, \cite{GK02}).

{\em BQFA as a generalization of other models.} Since unitary operations and orthogonal measurements are bistochastic operations, MO-BQFA is a generalization of LQFA and MM-BQFA is
a generalization of EQFA. Also one can see that MO-BQFA and MM-BQFA are generalizations of C-PRA and DH-PRA, respectively. A probability distribution vector $P=\sum\limits_ip_i|q_i\rangle$ of a PRA corresponds to the mixed state
$\rho=\sum\limits_ip_i|q_i\rangle\langle q_i|$ of a BQFA. Any transition matrix $B$ of a PRA is doubly stochastic. By the Birkhoff theorem \cite[Theorem 4.21]{Z99}, any doubly stochastic matrix is a convex combination of
some permutation matrices. Thus $B=\sum\limits_s p_sT_s$, where $p_s$ are nonnegative numbers with sum equal to $1$ and $T_s$ - permutation matrices.
So the CP bistochastic map corresponding to the transition matrix $B$ is $\Phi(\rho)=\sum\limits_s p_sT_s\rho T_s^*$, which is a random unitary operation.
Indeed, one may check that $\Phi(\rho)$ is a diagonal matrix such that  $(\Phi(\rho))_{ii}=(BP)_i$.

On the other hand, BQFA are a special case of one-way general QFA (also called quantum automata with open time evolution), which admit any CP trace-preserving transition maps. One-way general QFA recognize with
bounded error exactly the regular languages \cite{H10,LM10}, this fact was also mentioned in \cite[Introduction]{AW02}. Similar models of quantum automata which recognize any regular language
have been proposed in \cite{P00,C01,BP03,D03}. So the recognition power of BQFA is also limited to regular languages only.

{\em Comparison of the language classes.} Having a certain class of automata $\boldsymbol{\mathcal{A}}$, let
us denote by $\boldsymbol{\mathcal{L}}(\boldsymbol{\mathcal{A}})$
the respective class of languages. Thus
$\boldsymbol{\mathcal{L}}(\text{GA})$ $=$
$\boldsymbol{\mathcal{L}}(\text{MO-QFA})$ $=$
$\boldsymbol{\mathcal{G}}$, $\boldsymbol{\mathcal{L}}(\text{C-PRA})$
$=$ $\boldsymbol{\mathcal{L}}(\text{LQFA})$ $=$ $\boldsymbol{\mathcal{L}}(\text{MO-BQFA})$ $=$
$\boldsymbol{\mathcal{EJ}}$, $\boldsymbol{\mathcal{G}}$ $\subsetneq$
$\boldsymbol{\mathcal{L}}(\text{RFA})$ $\subsetneq$
$\boldsymbol{\mathcal{ER}_1}$, $\boldsymbol{\mathcal{EJ}}$
$\subsetneq$ $\boldsymbol{\mathcal{L}}(\text{MM-QFA})$
$\sabove{=}{?}$ $\boldsymbol{\mathcal{L}}(\text{DH-PRA})$
$\sabove{=}{?}$ $\boldsymbol{\mathcal{L}}(\text{EQFA})$ $\sabove{=}{?}$ $\boldsymbol{\mathcal{L}}(\text{MM-BQFA})$ $\subsetneq$
$\boldsymbol{\mathcal{ER}}$.
Relations concerning BQFA are proved below. All the other relations are known from the references given in Table \ref{table_automata}.
\begin{theorem}
\label{theor_EJ_MO-BQFA}
$\boldsymbol{\mathcal{L}}(\text{MO-BQFA})$ $=$ $\boldsymbol{\mathcal{EJ}}$.
\end{theorem}
\begin{proof}
Since $\boldsymbol{\mathcal{L}}(\text{LQFA})=\boldsymbol{\mathcal{EJ}}$ \cite{AT04} and MO-BQFA is a generalization of LQFA, $\boldsymbol{\mathcal{EJ}}\subseteq\boldsymbol{\mathcal{L}}(\text{MO-BQFA})$.
It remains to prove that $\boldsymbol{\mathcal{L}}(\text{MO-BQFA})\subseteq\boldsymbol{\mathcal{EJ}}$.

Suppose that a MO-BQFA $\mathcal{A}$ recognizes a language $L$ over alphabet $A$, such that $L\notin A^*\boldsymbol{\mathcal{EJ}}$. Let $\mathcal{M}=\mathcal{M}(L)$ - the syntactic monoid of $L$ and $\varphi$ - the
syntactic morphism from $A^*$ to $\mathcal{M}$. By assumption, there exist $x,y\in\mathcal{M}$ such that $(x^\omega y^\omega)^\omega\neq(y^\omega x^\omega)^\omega$.
There exists a positive integer $k$ such that for all $z$ in $\mathcal{M}$ $z^k=z^\omega$, therefore $(x^ky^k)^k\neq(y^kx^k)^k$. Let $\mathbf{a}\in x^k\varphi^{-1}$ and $\mathbf{b}\in y^k\varphi^{-1}$.
Consider the CP bistochastic transition maps $\Phi_\mathbf{a}$ and $\Phi_\mathbf{b}$ of  $\mathcal{A}$.
Theorem \ref{theor_bist_EJ} implies that there exists a sequence of positive integers $s_n$ such that
$\lim\limits_{n\to\infty}(\Phi_\mathbf{a}^{s_n}\circ\Phi_\mathbf{b}^{s_n})^n=\lim\limits_{n\to\infty}(\Phi_\mathbf{b}^{s_n}\circ\Phi_\mathbf{a}^{s_n})^n=(\Phi_\mathbf{a}^\omega\circ\Phi_\mathbf{b}^\omega)^\omega
=(\Phi_\mathbf{b}^\omega\circ\Phi_\mathbf{a}^\omega)^\omega$. Note that $\|\Phi_\mathbf{u}\|=\|\Phi_\mathbf{v}\|=1$. Therefore for any $\epsilon>0$ there exists $n>0$ such that for any $\mathbf{u},\mathbf{v}\in A^*$
$|p_{\mathbf{u}(\mathbf{a}^{s_n}\mathbf{b}^{s_n})^n\mathbf{v}}-p_{\mathbf{u}(\mathbf{b}^{s_n}\mathbf{a}^{s_n})^n\mathbf{v}}|<\epsilon$.
So there exists $n$ such that $(\mathbf{a}^{s_n}\mathbf{b}^{s_n})^n\sim_L(\mathbf{b}^{s_n}\mathbf{a}^{s_n})^n$.
Hence $(\mathbf{a}^{s_n}\mathbf{b}^{s_n})^n\varphi=(\mathbf{b}^{s_n}\mathbf{a}^{s_n})^n\varphi=(x^ky^k)^n=(y^kx^k)^n$. The latter implies $(x^ky^k)^k=(y^kx^k)^k$. This is a contradiction.
\qed
\end{proof}
The next theorem is equivalent to the statement that MM-BQFA can't recognize any language that does not satisfy the partial order condition from \cite{BP99}.
\begin{theorem}
\label{theor_MM_BQFA_subset_ER}
$\boldsymbol{\mathcal{L}}(\text{MM-BQFA})$ $\subseteq$ $\boldsymbol{\mathcal{ER}}$.
\end{theorem}
\begin{proof}
Suppose that a MM-BQFA $\mathcal{A}$ recognizes a language $L$ over alphabet $A$, such that $L\notin A^*\boldsymbol{\mathcal{ER}}$. Let $\mathcal{M}=\mathcal{M}(L)$ - the syntactic monoid of $L$ and $\varphi$ - the
syntactic morphism from $A^*$ to $\mathcal{M}$. By assumption, there exist $x,y\in\mathcal{M}$ such that $(x^\omega y^\omega)^\omega x^\omega\neq(x^\omega y^\omega)^\omega$.
There exists a positive integer $k$ such that for all $z$ in $\mathcal{M}$ $z^k=z^\omega$, therefore $(x^ky^k)^kx^k\neq(x^ky^k)^k$. Let $\mathbf{a}\in x^k\varphi^{-1}$ and $\mathbf{b}\in y^k\varphi^{-1}$.
Consider the CP sub-bistochastic maps $\Psi_\mathbf{a}$ and $\Psi_\mathbf{b}$ of $\mathcal{A}$.
Theorem \ref{theor_bist_EJ} implies that there exists a sequence of positive integers $s_n$ such that
$\lim\limits_{n\to\infty}\Psi_\mathbf{a}^{s_n}\circ(\Psi_\mathbf{b}^{s_n}\circ\Phi_\mathbf{a}^{s_n})^n=\lim\limits_{n\to\infty}(\Psi_\mathbf{b}^{s_n}\circ\Psi_\mathbf{a}^{s_n})^n
=(\Phi_\mathbf{b}^\omega\circ\Phi_\mathbf{a}^\omega)^\omega$. Let $\mathbf{u},\mathbf{v}\in A^*$. Let $\mathbf{a}(n)=\mathbf{a}^{s_n}$ and $\mathbf{b}(n)=\mathbf{b}^{s_n}$.
Let $\mathbf{w}(n)=\#\mathbf{u}(\mathbf{a}(n)\mathbf{b}(n))^n$. After reading the word $\mathbf{w}(n)$ the probability distribution is
$S_{\mathbf{w}(n)}=(\Psi_{\mathbf{w}(n)}(|q_0\rangle\langle q_0|),p_{acc},p_{rej})$. If $\mathbf{a}(n)$ is read afterwards,
$S_{\mathbf{w}(n)\mathbf{a}(n)}=(\Psi_{\mathbf{w}(n)\mathbf{a}(n)}(|q_0\rangle\langle q_0|),p_{acc}^\prime,p_{rej}^\prime)$,
where $p_{acc}^\prime\geqslant p_{acc}$ and $p_{rej}^\prime\geqslant p_{rej}$.
Note that $\lim\limits_{n\to\infty}\Psi_{\mathbf{w}(n)}=\lim\limits_{n\to\infty}\Psi_{\mathbf{w}(n)\mathbf{a}(n)}$ and $\|\Psi_\mathbf{u}\|\leqslant1$, $\|\Psi_\mathbf{v}\|\leqslant1$.
Therefore for any $\epsilon>0$ there exists $n>0$ such that for any $\mathbf{u}$ $\|\Psi_{\mathbf{w}(n)}(|q_0\rangle\langle q_0|)-\Psi_{\mathbf{w}(n)\mathbf{a}(n)}(|q_0\rangle\langle q_0|)\|<\epsilon$.
The latter in turn implies
$p_{acc}^\prime-p_{acc}<\epsilon$ and $p_{rej}^\prime-p_{rej}<\epsilon$.
So there exists $n$ such that for any $\mathbf{u},\mathbf{v}$ $\mathbf{u}(\mathbf{a}^{s_n}\mathbf{b}^{s_n})^n\mathbf{v}\in L$ if and only if
$\mathbf{u}(\mathbf{a}^{s_n}\mathbf{b}^{s_n})^n\mathbf{a}^{s_n}\mathbf{v}\in L$.
Hence $(\mathbf{a}^{s_n}\mathbf{b}^{s_n})^n\varphi=(\mathbf{a}^{s_n}\mathbf{b}^{s_n})^n\mathbf{a}^{s_n}\varphi=(x^ky^k)^n=(x^ky^k)^nx^k$. The latter implies $(x^ky^k)^k=(x^ky^k)^kx^k$. This is a contradiction.
\qed
\end{proof}
The relation $\boldsymbol{\mathcal{L}}(\text{MM-BQFA})$ $\subsetneq$ $\boldsymbol{\mathcal{ER}}$ is demonstrated in Section \ref{Lin_Ineq} (Corollary \ref{cor_subset}).

$\boldsymbol{\mathcal{L}}(\text{MM-BQFA})$ shares a lot of properties with the language classes of other "decide-and-halt" word acceptance models, like closure under complement and inverse homomorphisms.
In Section \ref{section_forb_constr} it is noted that MM-BQFA does not recognize any of the languages corresponding to "forbidden constructions" from \cite[Theorem 4.3]{AV00}.
Similarly as other "decide-and-halt" models, $\boldsymbol{\mathcal{L}}(\text{MM-BQFA})$ is not closed under union and intersection.
\begin{theorem}
\label{theor_closed_compl}
The class $\boldsymbol{\mathcal{L}}(\text{MM-BQFA})$ is closed under complement, inverse free monoid morphisms, and word quotient.
\end{theorem}
\begin{proof}
The proof goes along the same lines as in \cite[Theorem 4.1]{BP99}, where the same was proved for $\boldsymbol{\mathcal{L}}(\text{MM-QFA})$.
Closure under complement follows from the fact that we can exchange the accepting and rejecting states of the MM-BQFA.
Closure under inverse free monoid morphisms is proved in the same way as in \cite{BP99}, it is implied by the deferred measurement principle \cite[p.186]{NC00}.
Closure under word quotient is implied by closure under inverse free monoid morphisms and the presence of both end-markers.
\qed
\end{proof}
Non-closure under union and intersection is demonstrated in Section \ref{construct_DH_PRA} (Corollary \ref{cor_not_closed}).
\section{Linear Inequalities}\label{Lin_Ineq}
In this section, we  derive a system of linear inequalities that an $\mathcal{R}_1$ language recognized by a MM-BQFA must satisfy.
Let $\mathcal{S}$ be a MM-BQFA over alphabet $A$.
Let
$\{\mathbf{v_0},\mathbf{v_1},...,\mathbf{v_R}\}=\mathcal{F}(A)$. Assume $\mathbf{v_0}=\varepsilon$. For any $\mathbf{u}\in A^*$, let $\Psi(\mathbf{u})=\Psi_{\#\mathbf{u}}$.
Recall $\tau$ is the natural morphism from $A^*$ to $\mathcal{F}(A)$
(see Section \ref{sec_prelim} and subsection \ref{subsec_R1}). First, we prove that there exist words
$\mathbf{u_0},\mathbf{u_1},...,\mathbf{u_R}\in A^*$, for each $i$ $\mathbf{u_i}\tau=\mathbf{v_i}$, such that
the automaton $\mathcal{S}$ has essentially the same scaled density matrices for the words consisting of the same
letters:
\begin{proposition}
\label{nonh-theorem}
For every $\epsilon>0$ there exists an
everywhere defined injective function $\theta$ from $\mathcal{F}(A)$
to $A^*$ such that for all
$\mathbf{v},\mathbf{v^\prime}\in\mathcal{F}(A)$
\begin{conditions}
\item $\mathbf{v}\theta\tau=\mathbf{v}$;
\item $\mathbf{v}\leqslant\mathbf{v^\prime}$ if and only if $\mathbf{v}\theta\leqslant\mathbf{v^\prime}\theta$;
\item\label{item3} if $\mathbf{v}\sim_\omega\mathbf{v^\prime}$, then
$\|\Psi(\mathbf{v}\theta)-\Psi(\mathbf{v^\prime}\theta)\|<\epsilon$.
\end{conditions}
\end{proposition}
\begin{proof}
Let $m_l$ ($l=1,2,...$) be a sequence of positive integers such that for all letters $a\in A$ $\lim\limits_{l\to\infty}\Psi(a^{m_l})=\Psi_a^\omega$ (existence is implied by Theorem \ref{theor_idemp_exists}
and \cite[Theorem 201]{HW79}).

Let $\mu$ be a function that assigns to any
word in $A^*$ the same word (of the same length) with letters sorted
in alphabetical order.
Let $\varkappa_i$, $i\in\mathbb{N}$, a
morphism from $A^*$ to $A^*$ such that for any $a\in A$
$a\varkappa_i=a^i$.

Let $\xi=\xi_l$ be an everywhere defined function from
$\mathcal{F}(A)$ to $A^*$, such that $\varepsilon\xi=\varepsilon$ and
for all $\mathbf{v}\in\mathcal{F}(A)$, if $|\mathbf{v}|=1$ then
$\mathbf{v}\xi=\mathbf{v}^{m_l}$ and otherwise, if $|\mathbf{v}|\geqslant2$ then
$\mathbf{v}\xi=(\mathbf{v}\mu\varkappa_{m_l})^l$.

For any $\mathbf{v}$ in $\mathcal{F}(A)$, where
$\mathbf{v}=a_1...a_k$ ($a_i$ are distinct letters of $A$), define a finite
sequence of prefixes, denoted $\mathbf{v}[i]$, where
$\mathbf{v}[0]=\varepsilon$ and for all $i$, $1\leqslant i\leqslant
k$, $\mathbf{v}[i]=a_1...a_i$.

Let us define the function $\theta=\theta_l$ by induction as
follows. Let $\mathbf{v}[0]\theta=\varepsilon\theta=\varepsilon$ and
for all $i$, $1\leqslant i\leqslant k$, let
$\mathbf{v}[i]\theta=(\mathbf{v}[i-1]\theta)(\mathbf{v}[i]\xi)$.

So
$\mathbf{v}\theta=\mathbf{v}[k]\theta=\mathbf{v}[1]\xi...\mathbf{v}[k]
\xi=a_1^{m_l}((a_1^{m_l}a_2^{m_l})\mu)^l...((a_1^{m_l}a_2^{m_l}\dots
a_k^{m_l})\mu)^l$.
By construction, $(\mathbf{v}\theta)\tau=a_1a_2...a_k=\mathbf{v}$.

Consider $\mathbf{v},\mathbf{v^\prime}\in\mathcal{F}(A)$. Since
$\tau$ preserves order,
$\mathbf{v}\theta\leqslant\mathbf{v^\prime}\theta$ implies
$\mathbf{v}\leqslant\mathbf{v^\prime}$. Suppose
$\mathbf{v}\leqslant\mathbf{v^\prime}$. By construction,
$\mathbf{v}\theta\leqslant\mathbf{v^\prime}\theta$.

Suppose $\mathbf{v}\sim_\omega\mathbf{v^\prime}$. If
$|\mathbf{v}|=|\mathbf{v^\prime}|\leqslant1$, the condition
(\ref{item3}) of the proposition is satisfied. Hence assume
$|\mathbf{v}|=|\mathbf{v^\prime}|\geqslant2$. Theorem \ref{theor_bist_EJ}
implies that
$\lim\limits_{l\to\infty}\Psi(\mathbf{v}\theta)$ $=$
$\lim\limits_{l\to\infty}\Psi(\mathbf{v}\xi)$ and
$\lim\limits_{l\to\infty}\Psi(\mathbf{v^\prime}\theta)$ $=$
$\lim\limits_{l\to\infty}\Psi(\mathbf{v^\prime}\xi)$. Since
$\mathbf{v}$ is a permutation of $\mathbf{v^\prime}$,
$\lim\limits_{l\to\infty}\Psi(\mathbf{v}\xi)$ $=$
$\lim\limits_{l\to\infty}\Psi(\mathbf{v^\prime}\xi)$. Hence
$\lim\limits_{l\to\infty}\Psi(\mathbf{v}\theta)$ $=$
$\lim\limits_{l\to\infty}\Psi(\mathbf{v^\prime}\theta)$.
Take $\epsilon>0$. The last equality implies that there exists $l$
such that
$\|\Psi(\mathbf{v}\theta)-\Psi(\mathbf{v^\prime}\theta)\|<\epsilon$.
Since the monoid $\mathcal{F}(A)$ is finite, there exists $l$ which
satisfies the last inequality for any two words $\mathbf{v}$ and
$\mathbf{v^\prime}$ in $\mathcal{F}(A)$ such that
$\mathbf{v}\sim_\omega\mathbf{v^\prime}$.

Proposition is proved.\qed
\end{proof}
We are ready to derive the linear inequalities that must be satisfied by $\mathcal{S}$, if it recognizes an $\mathcal{R}_1$ language $L$ over
alphabet $A$.

Consider $\mathbf{x}\in A^*$. Suppose
$\mathbf{x}\tau=\mathbf{v}=a_1a_2...a_{|\mathbf{v}|}$.
By construction used in
the proof of Proposition \ref{nonh-theorem},
$\mathbf{v}\theta=(\mathbf{v}[1]\xi)(\mathbf{v}[2]\xi)...(\mathbf{v}\xi)$.
Now $\mathbf{x}\sim_\tau\mathbf{v}\sim_\tau\mathbf{v}\theta$.
Therefore $\mathbf{x}\in L$ if and only if $\mathbf{v}\in L$, and if
and only if $\mathbf{v}\theta\in L$.

Let us observe how $\mathcal{S}$ processes the input word
$\mathbf{v}\theta=(\mathbf{v}[1]\xi)(\mathbf{v}[2]\xi)...(\mathbf{v}\xi)$.
By the definition of MM-BQFA, any input word is enclosed by
end-markers $\#$ and $\$$. Let $r_0$ be the probability that
$\mathcal{S}$ has accepted the input (and halted) after reading the
initial end-marker $\#$. For $1\leqslant i\leqslant|\mathbf{v}|$,
let $r_{\mathbf{v}[i]}$ be the probability that $\mathcal{S}$ is in
a mixed state before reading the first letter of
$\mathbf{v}[i]\xi$ and has accepted the input (and halted) after reading
$\mathbf{v}[i]\xi$, including the possibility of halting while
reading it. Let $g_\mathbf{v}$ be the probability that $\mathcal{S}$
is in a mixed state after reading $\mathbf{v}\xi$ and has accepted the input after reading the final end-marker $\$$. It
follows that $\mathcal{S}$ accepts $\mathbf{v}\theta$ with
probability
$p_{\mathbf{v}\theta}=r_0+r_{\mathbf{v}[1]}+r_{\mathbf{v}[2]}+...+r_\mathbf{v}+g_\mathbf{v}$.
Note that the values $r_{\mathbf{v}[i]}$ and $g_\mathbf{v}$ depend on the chosen
function $\theta$, which itself depends on the parameter $l$.

We aim to prove that for any $\mathcal{R}_1$ language $L$ over alphabet $A$ it is possible to define a linear
system of inequalities $\mathfrak{L}$ such that the system is
consistent if and only if $L$ can be recognized by MM-BQFA. First of
all, it is necessary to define the system itself.

The probabilities $r_{\mathbf{v}[i]}$ can be regarded as symbolic variables
(let's call them s-variables) in the formal expression
$\widehat{p}_{\mathbf{v}\theta}=\widehat{r}_0+\widehat{r}_{\mathbf{v}[1]}+
\widehat{r}_{\mathbf{v}[2]}+...+\widehat{r}_\mathbf{v}+\widehat{g}_\mathbf{v}$.
\begin{definition}
Two s-variables $\widehat{r}_{\mathbf{v}[i]}$ and
$\widehat{r}_{\mathbf{v^\prime}[j]}$, $1\leqslant
i\leqslant|\mathbf{v}|$, $1\leqslant j\leqslant|\mathbf{v^\prime}|$,
are called {\em equivalent},
$\widehat{r}_{\mathbf{v}[i]}\sim\widehat{r}_{\mathbf{v^\prime}[j]}$,
if $\mathbf{v}[i-1]\sim_\omega\mathbf{v^\prime}[j-1]$ and
$\mathbf{v}[i]\sim_\omega\mathbf{v^\prime}[j]$. Two s-variables
$\widehat{g}_{\mathbf{v}}$ and $\widehat{g}_{\mathbf{v^\prime}}$ are
called {\em equivalent},
$\widehat{g}_{\mathbf{v}}\sim\widehat{g}_{\mathbf{v^\prime}}$, if
$\mathbf{v}\sim_\omega\mathbf{v^\prime}$.
\end{definition}
The s-variable $\widehat{r}_0$ is defined to be the only element of the equivalence class $[\widehat{r}_0]$.
The relation $\sim$ is an equivalence
relation in the two sets $\{\widehat{r}_{\mathbf{v}[i]}\ |\
\mathbf{v}\in\mathcal{F}(A)\ \text{and}\ 1\leqslant
i\leqslant|\mathbf{v}|\}$ and $\{\widehat{g}_{\mathbf{v}}\ |\
\mathbf{v}\in\mathcal{F}(A)\}$.

If two s-variables $\widehat{r}_{\mathbf{v}[i]}$ and
$\widehat{r}_{\mathbf{v^\prime}[j]}$ are equivalent then $i=j$.
Moreover, let's formulate the following
\begin{proposition}\label{prop_ineq_prob}
For any $\epsilon>0$ there exists a function $\theta$ from
Proposition \ref{nonh-theorem} such that for any $\mathbf{v},\mathbf{v^\prime}\in\mathcal{F}(A)$ and any prefixes $\mathbf{v}[i]$, $\mathbf{v^\prime}[i]$
\begin{conditions}
\item if
$\widehat{r}_{\mathbf{v}[i]}\sim\widehat{r}_{\mathbf{v^\prime}[i]}$
then $|r_{\mathbf{v}[i]}-r_{\mathbf{v^\prime}[i]}|<\epsilon$;
\item if $\widehat{g}_{\mathbf{v}}\sim\widehat{g}_{\mathbf{v^\prime}}$
then $|g_{\mathbf{v}}-g_{\mathbf{v^\prime}}|<\epsilon$.
\end{conditions}
\end{proposition}
\begin{proof}
Suppose $\mathbf{v}[i-1]\sim_\omega\mathbf{v^\prime}[i-1]$ and
$\mathbf{v}[i]\sim_\omega\mathbf{v^\prime}[i]$. In that case, $\mathbf{v}[i]\xi=\mathbf{v^\prime}[i]\xi$.
Proposition \ref{nonh-theorem} implies that for any $\epsilon^\prime$
$\|\Psi(\mathbf{v}[i-1]\theta)-\Psi(\mathbf{v^\prime}[i-1]\theta)\|<\epsilon^\prime$.
Hence after reading $\mathbf{v}[i-1]\theta$ or
$\mathbf{v^\prime}[i-1]\theta$ the automaton $\mathcal{S}$ comes to
essentially the same scaled mixed state. Within a particular step, the probability of accepting the input (and halting)
in the future depends only on the current mixed state and the
remaining part of the input word. So reading afterwards the word
$\mathbf{v}[i]\xi$, which is equal to $\mathbf{v^\prime}[i]\xi$,
implies that for any $\epsilon$
$|r_{\mathbf{v}[i]}-r_{\mathbf{v^\prime}[i]}|<\epsilon$.

Suppose $|\mathbf{v}|=|\mathbf{v^\prime}|$ and
$\mathbf{v}\sim_\omega\mathbf{v^\prime}$. Again, after reading the
both words $\mathbf{v}\theta$ and $\mathbf{v^\prime}\theta$ the
automaton $\mathcal{S}$ is in essentially the same scaled mixed state. So reading the final end-marker yields that for any $\epsilon$
$|g_\mathbf{v}-g_\mathbf{v^\prime}|<\epsilon$.\qed
\end{proof}
Recall that $\mathcal{F}(A)$ can be viewed as an automaton that
recognizes an $\mathcal{R}_1$ language $L$,
provided $L\tau$ is its set of final states. By Proposition
\ref{prop_ineq_prob}, all s-variables in the same equivalence class
may be replaced by a single variable. Now define a linear system of
inequalities $\mathfrak{L}$ as follows:
\begin{definition}
\label{def_inequalities}
The construction of the linear system of inequalities $\mathfrak{L}=\mathfrak{L}(L)$
for a given $\mathcal{R}_1$ language $L$.
\begin{conditions}
\item Take the formal expressions
$\widehat{p}_{\mathbf{v}\theta}=\widehat{r}_0+\widehat{r}_{\mathbf{v}[1]}+
\widehat{r}_{\mathbf{v}[2]}+...+\widehat{r}_\mathbf{v}+\widehat{g}_\mathbf{v}$
for all $\mathbf{v}\in\mathcal{F}(A)$;
\item Obtain linear expressions $\mathfrak{L}(\mathbf{v})$ from
$\{\widehat{p}_{\mathbf{v}\theta}\ |\ \mathbf{v}\in\mathcal{F}(A)\}$
in the following way; all s-variables in the same
equivalence class $[\widehat{r}]$ are replaced by a single variable denoted $\mathfrak{L}(\widehat{r})$, while any two
s-variables in different equivalence classes are replaced by different
variables;
\item Introduce yet another variables $p_1$ and $p_2$. For any $\mathbf{v}\in\mathcal{F}(A)$, if $\mathbf{v}\in
L\tau$, construct an inequality $\mathfrak{L}(\mathbf{v})\geqslant p_2$,
otherwise construct an inequality $\mathfrak{L}(\mathbf{v})\leqslant p_1$;
\item Append the system by an inequality $p_1<p_2$.
\end{conditions}
\end{definition}
If a MM-BQFA $\mathcal{S}$ recognizes an $\mathcal{R}_1$ language $L$, then the linear system of inequalities
$\mathfrak{L}$ is consistent. Thus we have established the following
result.
\begin{theorem}
\label{theor_sys_not_consist}
Suppose $L$ is an $\mathcal{R}_1$ language.
If the linear system $\mathfrak{L}$ is not consistent, then $L$
cannot be recognized by any MM-BQFA.
\end{theorem}
Therefore, if the linear system $\mathfrak{L}$ is not consistent, then $L$
cannot be recognized by any MM-QFA, DH-PRA or EQFA as well.
\begin{corollary}
\label{cor_subset}
$\boldsymbol{\mathcal{L}}(\text{MM-BQFA})$ $\subsetneq$ $\boldsymbol{\mathcal{ER}}$.
\end{corollary}
\begin{proof}
Consider an $\mathcal{R}_1$ language $L=\{\mathbf{ab,bac}\}$ over alphabet $A=\{a,b,c\}$.
Among others, the system $\mathfrak{L}$ has the following inequalities:
$$
\begin{array}{lclcrcl}
\mathfrak{L}(\mathbf{ab})&=&x_0+x_a+x_{ab}+y_{ab}&&&\geqslant&p_2\\
\mathfrak{L}(\mathbf{bac})&=&x_0+x_b+x_{ba}+x_{abc}&+&y_{abc}&\geqslant&p_2\\
\mathfrak{L}(\mathbf{ba})&=&x_0+x_b+x_{ba}+y_{ab}&&&\leqslant&p_1\\
\mathfrak{L}(\mathbf{abc})&=&x_0+x_a+x_{ab}+x_{abc}&+&y_{abc}&\leqslant&p_1\\
&&&&p_1&<&p_2
\end{array}
$$
The above inequalities define a system that is not consistent. Hence $\mathfrak{L}$ is not consistent as well. So by Theorem \ref{theor_sys_not_consist} $L$ cannot be recognized by any MM-BQFA.
Therefore Theorem \ref{theor_MM_BQFA_subset_ER} implies that $\boldsymbol{\mathcal{L}}(\text{MM-BQFA})$ $\subsetneq$ $\boldsymbol{\mathcal{ER}}$.
\qed
\end{proof}
To prove the statement converse to Theorem \ref{theor_sys_not_consist}, we need to indicate some of the
properties of the obtained system $\mathfrak{L}$. The converse
statement itself will be proved in Section \ref{construct_DH_PRA} (Theorem \ref{theor_sys_consist}).

Consider the inequalities in the system. Let $y_A$ be the unique variable $\mathfrak{L}(\widehat{g}_\mathbf{w})$, such that $\mathbf{w}\omega=A$.
Except for the inequality $p_1<p_2$, the left-hand side of any inequality has the
form $\mathfrak{L}(\mathbf{v})=x_0+x_{\mathbf{v}[1]}+
x_{\mathbf{v}[2]}+...+x_\mathbf{v}+y_\mathbf{v}$, where
$\mathbf{v}\in\mathcal{F}(A)$, $x_0=\mathfrak{L}(\widehat{r}_0)$,
$x_{\mathbf{v}[j]}=\mathfrak{L}(\widehat{r}_{\mathbf{v}[j]})$, $1\leqslant j\leqslant|\mathbf{v}|$, and
$y_\mathbf{v}=\mathfrak{L}(\widehat{g}_\mathbf{v})$.

The only possible coefficients of variables in any linear
inequality are $-1$, $0$ and $1$. Denote by
$Z=\{x_0,z_1,...,z_s,y_1,...,y_t,p_1,p_2\}$ the set of all the
variables in the system $\mathfrak{L}$, where $z_i$ are variables of
the form $x_{\mathbf{v}[j]}$, and $y_i$ are variables of the form
$y_\mathbf{v}$. Denote by $N$ the total number of variables.

Let $M=|A|+2$, which is the maximal number of variables (with
nonzero coefficients) in any expression $\mathfrak{L}(\mathbf{v})$.
Each expression $\mathfrak{L}(\mathbf{v})$ has exactly one variable
$y_i$. If two expressions $\mathfrak{L}(\mathbf{v})$ and
$\mathfrak{L}(\mathbf{v^\prime})$ share the same variable $y_i$,
then $\mathbf{v}\sim_\omega\mathbf{v^\prime}$, so
$\mathfrak{L}(\mathbf{v})$ and $\mathfrak{L}(\mathbf{v^\prime})$
have the same number of variables. So it is possible to denote by
$n(y_i)$ the number of variables in any corresponding expression
$\mathfrak{L}(\mathbf{v})$. Let $d(y_i)=M-n(y_i)+1$.
\begin{proposition}\label{lemma_positive_solutions}
The system $\mathfrak{L}$ is consistent if and only if it has a
solution where all the variables are assigned nonnegative real
values.
\end{proposition}
\begin{proof}
Let $c_0,...,c_{N-1}$ be some real numbers. Let $C$ be any real
constant. Any inequality in the system can be written in one of the
three forms, namely,
\begin{eqnarray}
x_0+x_{\mathbf{v}[1]}+x_{\mathbf{v}[2]}+...+x_\mathbf{v}+y_\mathbf{v}&\geqslant&p_2,\\
x_0+x_{\mathbf{v}[1]}+x_{\mathbf{v}[2]}+...+x_\mathbf{v}+y_\mathbf{v}&\leqslant&p_1,\\
p_1&<&p_2.
\end{eqnarray}
The inequalities above are satisfied if and only the following
inequalities are satisfied;
\begin{eqnarray}
x_0+x_{\mathbf{v}[1]}+x_{\mathbf{v}[2]}+...+x_\mathbf{v}+y_\mathbf{v}+CM&\geqslant&p_2+CM,\\
x_0+x_{\mathbf{v}[1]}+x_{\mathbf{v}[2]}+...+x_\mathbf{v}+y_\mathbf{v}+CM&\leqslant&p_1+CM,\\
p_1+CM&<&p_2+CM.
\end{eqnarray}
Note that
\begin{eqnarray*}
x_0+x_{\mathbf{v}[1]}+x_{\mathbf{v}[2]}+...+x_\mathbf{v}+y_\mathbf{v}+CM=
(x_0+C)+(x_{\mathbf{v}[1]}+C)+\\+(x_{\mathbf{v}[2]}+C)+...
+(x_\mathbf{v}+C)+(y_\mathbf{v}+Cd(y_\mathbf{v}))
\end{eqnarray*}
Therefore the system $\mathfrak{L}$ has a solution
\begin{equation}\label{eq_sol}
\begin{array}{rll}
\{x_0&=&c_0,\\
z_1&=&c_1,\dots,\ z_s=c_s,\\
y_1&=&c_{s+1},\dots,\ y_t=c_{N-3},\\
p_1&=&c_{N-2},\ p_2=c_{N-1}\}
\end{array}
\end{equation}
if and only if it has a solution
\begin{equation}\label{eq_sol1}
\begin{array}{rll}
\{x_0&=&c_0+C,\\
z_1&=&c_1+C,\dots,\
z_s=c_s+C,\\
y_1&=&c_{s+1}+Cd(y_1),\dots,\ y_t=c_{N-3}+Cd(y_t),\\
p_1&=&c_{N-2}+CM,\ p_2=c_{N-1}+CM\}.
\end{array}
\end{equation}

Suppose the system $\mathfrak{L}$ is consistent and has a solution
(\ref{eq_sol}). Let $c_{\min}=\min\{0,c_0,...,c_{N-1}\}$. Take
$C=-c_{\min}$. By construction, $C$ is a nonnegative real number,
such that for all $i$, $c_i+C$ is also nonnegative. Now
(\ref{eq_sol1}) is the solution of the system $\mathfrak{L}$ such
that all the variables are assigned nonnegative values.\qed
\end{proof}
\begin{proposition}
\label{prop_sol_0_y_v}
The system $\mathfrak{L}$ is consistent if and only if it has a
solution where all the variables are assigned nonnegative real values and $x_0=0$, $y_A=0$.
\end{proposition}
\begin{proof}
Suppose $\mathfrak{L}$ is consistent. By Proposition \ref{lemma_positive_solutions}, the system has a solution (\ref{eq_sol}), where for all $i$ $c_i\geqslant0$.
We first prove that there exists a solution where $x_0=0$.

If $L\neq A^*$, there exists $\mathbf{v}\in\mathcal{F}(A)$ such that the inequality $\mathfrak{L}(\mathbf{v})\leqslant p_1$ is part of the system $\mathfrak{L}$.
Since $x_0$ is part of $\mathfrak{L}(\mathbf{v})$, $c_0\leqslant c_{N-2}<c_{N-1}$. So the system has a solution $\{x_0=0,z_1=c_1,\dots,z_s=c_s,y_1=c_{s+1},\dots,y_t=c_{N-3},
p_1=c_{N-2}^\prime,p_2=c_{N-1}^\prime\}$, where $c_{N-2}^\prime=c_{N-2}-c_0$ and $c_{N-1}^\prime=c_{N-1}-c_0$.
Otherwise, if $L=A^*$, take the solution $\{x_0=0,z_1=0,\dots,z_s=0,y_1=1,\dots,y_t=1,p_1=0,p_2=1\}$.

Next, we prove that there exists a solution where $x_0=0$, $y_A=0$.
The left-hand side of any inequality that contains the variable $y_A$ is of the form $x_0+x_{\mathbf{v}[1]}+x_{\mathbf{v}[2]}+...+x_\mathbf{v}+y_A$, where $\mathbf{v}\omega=A$.
Any inequality in the system either contains the variable $y_A$ and a single variable $x_\mathbf{v}$, such that $\mathbf{v}\omega=A$, or contains none of them.
Assume that $y_t$ is the variable $y_A$ and $z_{s-|A|+1},...,z_s$ are all the variables of the form $x_\mathbf{v}$, such that $\mathbf{v}\omega=A$.
Since the system has a solution $\{x_0=0,z_1=c_1,\dots,z_{s-|A|}=c_{s-|A|},z_{s-|A|+1}=c_{s-|A|+1},\dots,z_s=c_s,y_1=c_{s+1},\dots,y_{t-1}=c_{N-4},y_t=c_{N-3},p_1=c_{N-2}^\prime,p_2=c_{N-1}^\prime\}$, it also
has a solution where $z_{s-|A|+1}=c_{s-|A|+1}+c_{N-3},\dots,z_s=c_s+c_{N-3},y_t=0$ and other variables keep their previous values.
\qed
\end{proof}
\begin{proposition}
\label{prop_sol_0_1}
The system $\mathfrak{L}$ is consistent if and only if it has a
solution where all the variables are assigned real values from $0$ to $1$ and $x_0=0$, $y_A=0$.
\end{proposition}
\begin{proof}
Suppose $\mathfrak{L}$ is consistent. By Proposition
\ref{prop_sol_0_y_v}, the system has a solution
(\ref{eq_sol}), where for all $i$ $c_i\geqslant0$ and $x_0=0$, $y_A=0$. Assume that $y_t$ is the variable $y_A$. Let
$D=\max\{c_i\}$. Since $p_1<p_2$, $D>0$. So the solution (\ref{eq_sol}) may be divided by $D$ and the system
$\mathfrak{L}$ has a solution
\begin{equation}\label{eq_sol2}
\begin{array}{rll}
\{x_0&=&0,\\
z_1&=&c_1/D,\dots,\ z_s=c_s/D,\\
y_1&=&c_{s+1}/D,\dots,\ y_{t-1}=c_{N-4}/D,\ y_A=0\\
p_1&=&c_{N-2}/D,\ p_2=c_{N-1}/D\}.
\end{array}
\end{equation}
The solution (\ref{eq_sol2}) assigns to all the variables real
values from $0$ to $1$.\qed
\end{proof}
\begin{proposition}
\label{prop_sol_0_A}
The system $\mathfrak{L}$ is consistent if and only if it has a
solution where  $x_0=0$,  $y_A=0$,  $0\leqslant p_1,p_2\leqslant1$ and all the other variables $z_1,...,z_s,y_1,...,$ $y_{t-1}$ are assigned real values from $0$ to $1/|A|$.
\end{proposition}
\begin{proof}
Suppose $\mathfrak{L}$ is consistent. By Proposition \ref{prop_sol_0_1}, the system has a solution (\ref{eq_sol2}), where for all $i$ $0\leqslant c_i\leqslant1$. For any $i$, let $c_i^\prime=c_i/D$.
Let $c=\max\{c_1^\prime,...,c_{N-4}^\prime\}$.  If $c_{N-1}^\prime<c|A|$ then the solution is divided by $c|A|$. Otherwise, if $c_{N-1}^\prime\geqslant c|A|$
then $c\leqslant1/|A|$ and no scaling is necessary.
\qed
\end{proof}
\section{Construction of DH-PRA for $\mathcal{R}_1$ languages}\label{construct_DH_PRA}
In this section, a method will be provided that allows to construct a DH-PRA for any $\mathcal{R}_1$ language $L$ that generates a consistent system of linear inequalities.
Since MM-BQFA is a generalization of DH-PRA, this implies the construction of MM-BQFA as well. Recall $\sigma$ is a natural morphism from $\mathcal{F}(A)$ to $\mathcal{P}(A)$, defined in subsection \ref{subsec_R1}.

{\em Preparation of a linear programming problem.} Consider an $\mathcal{R}_1$ language $L$ over alphabet $A$. Construct the respective system of linear inequalities $\mathfrak{L}$.
Obtain a system $\mathfrak{L}_1$ by supplementing $\mathfrak{L}$ with additional inequalities that enforce the constraints expressed in Proposition \ref{prop_sol_0_A}, according to which $\mathfrak{L}$ is consistent
if and only if $\mathfrak{L}_1$ is consistent. Obtain a system $\mathfrak{L}_1^\prime$ by replacing in $\mathfrak{L}_1$ the inequality $p_1<p_2$ by $p_1\leqslant p_2$.
The linear programming problem, denoted $\mathfrak{P}$, is to maximize $p_2-p_1$ according to the constraints expressed by $\mathfrak{L}_1^\prime$.

Since $\mathfrak{L}_1^\prime$ is homogenous, it always has a solution where $p_1=p_2$. Since the solution polytope of $\mathfrak{L}_1^\prime$ is bounded, $\mathfrak{P}$ always has an optimal solution.
Obviously, if the optimal solution yields $p_1=p_2$, then $\mathfrak{L}_1$ is not consistent and therefore, by Theorem \ref{theor_sys_not_consist}, a DH-PRA that recognizes $L$ does not exist.
Otherwise, if the optimal solution yields $p_1<p_2$, then $\mathfrak{L}_1$ is consistent.

{\em Automata derived from the free semilattice $\mathcal{P}(A)$.} Assume $\mathfrak{L}_1$ is consistent, so we are able to obtain a solution of $\mathfrak{P}$ where $p_1<p_2$.
Given any expression $Z$ of variables from $\mathfrak{L}_1$, let $\mathfrak{P}(Z)$ - the value which is assigned to $Z$ by solving $\mathfrak{P}$.
First, we use the obtained solution to construct probabilistic automata $\mathcal{A}_i$, $1\leqslant i\leqslant|A|$. Those automata are not probabilistic reversible.
Similarly as in the "decide-and-halt" model, the constructed automata have accepting, rejecting and non-halting states. Any input word is appended by the end-marker $\$$.
The initial end-marker $\#$ is not used for those automata themselves.
Any automaton $\mathcal{A}_i$ is a tuple $(Q_i,A\cup\{\$\},s_i,\delta_i)$, where $Q_i$ is a set of states,
$s_i$ - an initial state and $\delta_i$ - a transition function $Q\times A\times Q\longrightarrow[0,1]$, so $\delta_i(q,a,q^\prime)$ is a probability of transit from $q$ to $q^\prime$ on reading input letter $a$.
\begin{figure}[b!]
    \begin{center}
        \unitlength=4pt
        \begin{picture}(70,54)(6.5,-4)
        \thinlines
        \gasset{Nw=3,Nh=3,Nmr=3,curvedepth=0}
        \node[Nmarks=i,iangle=90](S)(34.75,45){}
        \node(E1)(7,38){$_{\{\}}$}
        \node(Z1)(0,25){$_{acc}$}
        \node(A1)(4.67,25){$_{acc}$}
        \node(B1)(9.33,25){$_{acc}$}
        \node(C1)(14,25){$_{acc}$}
        \drawedge[ELside=r,ELdist=0,ELpos=40,curvedepth=0](E1,Z1){$_{\$,v_\varepsilon}$}
        \drawedge[ELside=l,ELdist=0.1,ELdistC=y,ELpos=50,curvedepth=0](E1,A1){$_{a,t_a}$}
        \drawedge[ELside=l,ELdist=0.2,ELdistC=y,ELpos=70,curvedepth=0](E1,B1){$_{b,t_b}$}
        \drawedge[ELside=l,ELdist=0,ELpos=40,curvedepth=0](E1,C1){$_{c,t_c}$}
        \gasset{Nw=3,Nh=3,Nmr=3,curvedepth=0}
        \node(E2)(34.75,38){$_{\{\}}$}
        \node(A2)(21.875,25){$_{\{a\}}$}
        \node(B2)(35,25){$_{\{b\}}$}
        \node(C2)(48.125,25){$_{\{c\}}$}
        \node(ZA2)(17.5,10){$_{acc}$}
        \node(AB2)(21.875,10){$_{acc}$}
        \node(AC2)(26.25,10){$_{acc}$}
        \node(ZB2)(30.625,10){$_{acc}$}
        \node(BA2)(35,10){$_{acc}$}
        \node(BC2)(39.375,10){$_{acc}$}
        \node(ZC2)(43.75,10){$_{acc}$}
        \node(CA2)(48.125,10){$_{acc}$}
        \node(CB2)(52.5,10){$_{acc}$}
        \drawloop[ELdist=0,loopdiam=1,loopangle=135](A2){$_a$}
        \drawloop[ELdist=0,loopdiam=1,loopangle=45](B2){$_b$}
        \drawloop[ELdist=0,loopdiam=1,loopangle=45](C2){$_c$}
        \drawedge[ELdist=0,ELside=r,curvedepth=0](E2,A2){$_a$}
        \drawedge[ELdist=0,ELpos=43,ELside=r,curvedepth=0](E2,B2){$_b$}
        \drawedge[ELdist=0,ELside=l,curvedepth=0](E2,C2){$_c$}
        \drawedge[ELpos=40,ELdist=0,ELdistC=n,ELside=r,curvedepth=0](A2,ZA2){$_{\$,v_a}$}
        \drawedge[ELpos=60,ELdist=0.6,ELdistC=y,ELside=l,curvedepth=0](A2,AB2){$_{b,t_{ab}}$}
        \drawedge[ELpos=40,ELdist=0,curvedepth=0](A2,AC2){$_{c,t_{ac}}$}
        \drawedge[ELpos=40,ELdist=0,ELdistC=n,ELside=r,curvedepth=0](B2,ZB2){$_{\$,v_b}$}
        \drawedge[ELpos=60,ELdist=0.5,ELdistC=y,ELside=l,curvedepth=0](B2,BA2){$_{a,t_{ba}}$}
        \drawedge[ELpos=40,ELdist=0,ELside=l,curvedepth=0](B2,BC2){$_{c,t_{bc}}$}
        \drawedge[ELpos=40,ELdist=0,ELdistC=n,ELside=r,curvedepth=0](C2,ZC2){$_{\$,v_c}$}
        \drawedge[ELpos=60,ELdist=0.5,ELdistC=y,ELside=l,curvedepth=0](C2,CA2){$_{a,t_{ca}}$}
        \drawedge[ELpos=40,ELside=l,ELdist=0,curvedepth=0](C2,CB2){$_{b,t_{cb}}$}
        \gasset{Nw=3,Nh=3,Nmr=3,curvedepth=0}
        \node(E3)(69,38){$_{\{\}}$}
        \node(A3)(59,25){$_{\{a\}}$}
        \node(B3)(69,25){$_{\{b\}}$}
        \node(C3)(79,25){$_{\{c\}}$}
        \gasset{Nw=4.5,Nh=4.5,Nmr=4,curvedepth=0}
        \node(AB3)(59,10){$_{\{a,b\}}$}
        \node(AC3)(69,10){$_{\{a,c\}}$}
        \node(BC3)(79,10){$_{\{b,c\}}$}
        \gasset{Nw=3,Nh=3,Nmr=3,curvedepth=0}
        \node(ZAB3)(56.5,-5){$_{acc}$}
        \node(ABC3)(61.5,-5){$_{acc}$}
        \node(ZAC3)(66.5,-5){$_{acc}$}
        \node(ACB3)(71.5,-5){$_{acc}$}
        \node(ZBC3)(76.5,-5){$_{acc}$}
        \node(BCA3)(81.5,-5){$_{acc}$}
        \drawloop[ELdist=0,loopdiam=1,loopangle=135](A3){$_a$}
        \drawloop[ELdist=0,loopdiam=1,loopangle=45](B3){$_b$}
        \drawloop[ELdist=0,loopdiam=1,loopangle=45](C3){$_c$}
        \drawloop[ELdist=0,loopdiam=1,loopangle=-135](AB3){$_{a,b}$}
        \drawloop[ELdist=0,loopdiam=1,loopangle=-135](AC3){$_{a,c}$}
        \drawloop[ELdist=0,loopdiam=1,loopangle=-135](BC3){$_{b,c}$}
        \drawedge[ELdist=0,ELside=r,curvedepth=0](E3,A3){$_a$}
        \drawedge[ELdist=0,ELpos=43,ELside=r,curvedepth=0](E3,B3){$_b$}
        \drawedge[ELdist=0,ELside=l,curvedepth=0](E3,C3){$_c$}
        \drawedge[ELdist=0,ELside=r,curvedepth=0](A3,AB3){$_b$}
        \drawedge[ELdist=0,ELpos=30,ELside=r,curvedepth=0](A3,AC3){$_c$}
        \drawedge[ELdist=0,ELpos=30,ELside=r,curvedepth=0](B3,AB3){$_a$}
        \drawedge[ELdist=0,ELpos=30,ELside=l,curvedepth=0](B3,BC3){$_c$}
        \drawedge[ELdist=0,ELpos=30,ELside=l,curvedepth=0](C3,AC3){$_a$}
        \drawedge[ELdist=0,ELside=l,curvedepth=0](C3,BC3){$_b$}
        \drawedge[ELpos=40,ELside=l,ELdist=0,ELdistC=n,curvedepth=0](AB3,ABC3){$_{c,t_{abc}}$}
        \drawedge[ELpos=60,ELside=r,ELdist=0,ELdistC=n,curvedepth=0](AB3,ZAB3){$_{\$,v_{ab}}$}
        \drawedge[ELpos=40,ELdist=0,ELside=l,ELdistC=n,curvedepth=0](AC3,ACB3){$_{b,t_{acb}}$}
        \drawedge[ELpos=60,ELside=r,ELdist=0,ELdistC=n,curvedepth=0](AC3,ZAC3){$_{\$,v_{ac}}$}
        \drawedge[ELpos=40,ELside=l,ELdist=0,curvedepth=0](BC3,BCA3){$_{a,t_{bca}}$}
        \drawedge[ELpos=60,ELside=r,ELdist=0,ELdistC=n,curvedepth=0](BC3,ZBC3){$_{\$,v_{bc}}$}
        \drawedge[ELdist=0,ELside=r,curvedepth=0](S,E1){$_{\#,1/3}$}
        \drawedge[ELpos=40,ELdist=0,ELside=l,curvedepth=0](S,E2){$_{\#,1/3}$}
        \drawedge[ELpos=40,ELdist=0,ELside=l,curvedepth=0](S,E3){$_{\#,1/3}$}
        \end{picture}
    \end{center}
    \caption{An automaton $\mathcal{A}$ over alphabet $\{a,b,c\}$, the rejecting states are not shown.}\label{fig.C}
\end{figure}
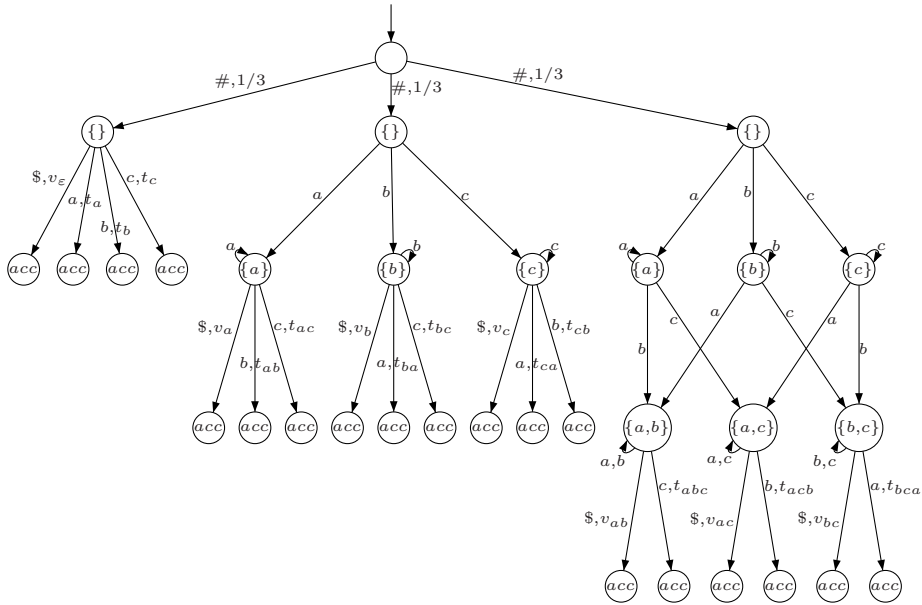
$\mathcal{A}_i$ is constructed as follows.
\begin{conditions}
\item Take the deterministic automaton $(\mathcal{P}(A),A,\emptyset,\cdp)$, remove all the states at level greater or equal to $i$. The remaining states are defined to be non-halting.
The state $\emptyset$ is initial, it is the only state of $\mathcal{A}_i$ at level $0$. For any  $a$ in $A$ and state $\mathbf{s}$ at levels $\{0,\dots,i-2\}$, $\delta_i(\mathbf{s},a,\mathbf{s}\cdp a)=1$.
For any state $\mathbf{s}$ at level $i-1$ and any $a$ in $\mathbf{s}$, $\delta_i(\mathbf{s},a,\mathbf{s})=1$.
\item \label{item_construct_DH_PRA_Ai_5} For any non-halting state $\mathbf{s}$ at levels $\{0,\dots,i-2\}$, add a rejecting state $(\mathbf{s}\$)_{rej}$.
Let $\delta_i(\mathbf{s},\$,(\mathbf{s}\$)_{rej})=1$.
\item For any state $\mathbf{s}$ at level $i-1$, add $|A|-|\mathbf{s}|+1$ accepting states $(\mathbf{s}a)_{acc}$, $a\in(A\setminus\mathbf{s})\cup\{\$\}$.
Also add $|A|-|\mathbf{s}|+1$ rejecting states $(\mathbf{s}a)_{rej}$, $a\in(A\setminus\mathbf{s})\cup\{\$\}$.
\item If $a\in A\setminus\mathbf{s}$, any element $\mathbf{s^\prime}a$ in $(\mathbf{s}\sigma^{-1})a$
defines s-variables in the same equivalence class $[\widehat{r}_{\mathbf{s^\prime}a}]$. Let $x_{\mathbf{s^\prime}a}=\mathfrak{L}(\widehat{r}_{\mathbf{s^\prime}a})$ and $c_{\mathbf{s^\prime}a}=\mathfrak{B}(x_{\mathbf{s^\prime}a})$.
Any element $\mathbf{s^\prime}$ in $\mathbf{s}\sigma^{-1}$ defines s-variables in the same equivalence class $[\widehat{g}_\mathbf{s^\prime}]$.
Let $y_\mathbf{s^\prime}=\mathfrak{L}(\widehat{g}_\mathbf{s^\prime})$
and $d_\mathbf{s^\prime}=\mathfrak{B}(y_\mathbf{s^\prime})$.
\item Define missing transitions for the states at level $i-1$.
\label{item_construct_DH_PRA_Ai_4} For any state $\mathbf{s}$ at level $i-1$ and any $a$ in $A\setminus\mathbf{s}$, let $t_{\mathbf{s^\prime}a}=\delta_i(\mathbf{s},a,(\mathbf{s}a)_{acc})=c_{\mathbf{s^\prime}a}|A|$
and $\delta_i(\mathbf{s},a,(\mathbf{s}a)_{rej})=1-t_{\mathbf{s^\prime}a}$.
Let
$v_\mathbf{s^\prime}=\delta_i(\mathbf{s},\$,(\mathbf{s}\$)_{acc})=d_\mathbf{s^\prime}|A|$ and $\delta_i(\mathbf{s},\$,(\mathbf{s}\$)_{rej})=1-v_\mathbf{s^\prime}$.
\item \label{item_construct_DH_PRA_Ai_6} Formally, we would need the transitions outgoing the halting states, those are left undefined.
\end{conditions}
Consider an automaton $\mathcal{A}$ (Figure \ref{fig.C}), which with the same probability $1/|A|$ executes any of the automata $\mathcal{A}_1,...,\mathcal{A}_{|A|}$
(i.e., it uses the initial end-marker $\#$ to transit to initial states of any of those automata).
By construction of $\mathcal{A}_1,...,\mathcal{A}_{|A|}$, the automaton $\mathcal{A}$ accepts any word $\mathbf{u}\in A^*$ with probability
$\mathfrak{P}(\mathfrak{L}(\mathbf{u}\tau))$. Since for any word $\mathbf{u}\in L$, $\mathfrak{P}(\mathfrak{L}(\mathbf{u}\tau))\geqslant\mathfrak{P}(p_2)$, and for any word $\mathbf{w}\notin L$,
$\mathfrak{P}(\mathfrak{L}(\mathbf{w}\tau))\leqslant\mathfrak{P}(p_1)$, the automaton $\mathcal{A}$ recognizes the language $L$.

{\em Construction of a DH-PRA.} In order to construct a DH-PRA recognizing $L$, it remains to demonstrate that any of the automata $\mathcal{A}_1,...,\mathcal{A}_{|A|}$
may be simulated by some DH probabilistic reversible automata, that is,
for any automaton $\mathcal{A}_i$, it is possible to construct a sequence of DH-PRA $\mathcal{S}_{i,n}$, where $n\geqslant1$,
such that $p_{\mathbf{w},\mathcal{S}_{i,n}}$ converges uniformly to $p_{\mathbf{w},\mathcal{A}_i}$ on $A^*$ as $n\to\infty$.

An automaton $\mathcal{A}_i=(Q_i,A\cup\{\$\},s_i,\delta_i)$ is used to construct
a DH-PRA $\mathcal{S}_{i,n}=(Q_{i,n},A\cup\{\$\},s_i,\delta_{i,n})$ as described next. Initially $Q_{i,n}$ is empty. Do the following.
\begin{conditions}
\item For any non-halting state $\mathbf{s}$ at level $j$, $0\leqslant j\leqslant i-1$, supplement $\mathcal{S}_{i,n}$ with non-halting states denoted $\mathbf{s}_k$, where $1\leqslant k\leqslant n^j$.
\item For any non-halting state $\mathbf{s}$ at level $j$, $0\leqslant j<i-1$, supplement $\mathcal{S}_{i,n}$ with rejecting states $(\mathbf{s}\$)_{rej,k}$, where $1\leqslant k\leqslant n^j$.
\item For any non-halting state $\mathbf{s}$ at level $i-1$, accepting state $(\mathbf{s}a)_{acc}$ and rejecting state $(\mathbf{s}a)_{rej}$, where $a\in(A\setminus\mathbf{s})\cup\{\$\}$,
supplement $\mathcal{S}_{i,n}$ with accepting states $(\mathbf{s}a)_{acc,k}$ and rejecting states $(\mathbf{s}a)_{rej,k}$, where $1\leqslant k\leqslant n^{i-1}$.
\end{conditions}
It remains to define the transitions. For any non-halting state $\mathbf{s}$ of $\mathcal{A}_i$ at level $j$, $1\leqslant j\leqslant i-1$, the states in
$\{\mathbf{s}_k\}$ are grouped into $n^{j-1}$ disjoint subsets with $n$ states in each,
so that any state in $\{\mathbf{s}_k\}$ may be denoted as $\mathbf{s}_{l,m}$, where $1\leqslant l\leqslant n^{j-1}$ and $1\leqslant m\leqslant n$.

For any letter $a$ in $A$, consider all pairs of non-halting states $\mathbf{s,t}$ of $\mathcal{A}_i$ such that $\mathbf{s}\neq\mathbf{t}$ and $\delta_i(\mathbf{s},a,\mathbf{t})=1$.
For any fixed $k$ and any $l$ and $m$, $1\leqslant l,m\leqslant n$, define
$\delta_{i,n}(\mathbf{s}_k,a,\mathbf{s}_k)=\delta_{i,n}(\mathbf{s}_k,a,\mathbf{t}_{k,m})=\delta_{i,n}(\mathbf{t}_{k,m},a,\mathbf{s}_k)=\delta_{i,n}(\mathbf{t}_{k,l},a,\mathbf{t}_{k,m})=1/(n+1)$.

For any non-halting state $\mathbf{s}$ of $\mathcal{A}_i$ at level $j$, $0\leqslant j<i-1$, $\delta_{i,n}(\mathbf{s}_k,\$,(\mathbf{s}\$)_{rej,k})=1$, $\delta_{i,n}((\mathbf{s}\$)_{rej,k},\$,\mathbf{s}_k)=1$.
For the same $(\mathbf{s}\$)_{rej,k}$ and any
other letter $b$ in $A\cup\{\$\}$, define $\delta_{i,n}((\mathbf{s}\$)_{rej,k},b,(\mathbf{s}\$)_{rej,k})=1$.

For any non-halting state $\mathbf{s}$ of $\mathcal{A}_i$ at level $i-1$ and $a\in(A\setminus\mathbf{s})\cup\{\$\}$, the transitions induced by $a$ among $\mathbf{s}_k$, $(\mathbf{s}a)_{acc,k}$, $(\mathbf{s}a)_{rej,k}$
are defined by the matrix
$\begin{pmatrix}
0&0&1\\
r_1&r_2&0\\
r_2&r_1&0\\
\end{pmatrix}$,
where $r_1=\delta_i(\mathbf{s},a,(\mathbf{s}a)_{acc})$, $r_2=\delta_i(\mathbf{s},a,(\mathbf{s}a)_{rej})$. The first, second and third rows and columns are indexed by
$\mathbf{s}_k$, $(\mathbf{s}a)_{acc,k}$, $(\mathbf{s}a)_{rej,k}$, respectively. Note that $r_1+r_2=1$.
For the same $(\mathbf{s}a)_{acc,k},(\mathbf{s}a)_{rej,k}$ and any
other letter $b$ in $A\cup\{\$\}$, define $\delta_{i,n}((\mathbf{s}a)_{acc,k},b,(\mathbf{s}a)_{acc,k})=\delta_{i,n}((\mathbf{s}a)_{rej,k},b,(\mathbf{s}a)_{rej,k})=1$.

We have defined all the non-zero transitions for $\mathcal{S}_{i,n}$. By construction, the transition matrices induced by any letter $a$ in $A\cup\{\$\}$ are doubly stochastic.
\begin{lemma}
\label{lemma_pra_unif_conv}
For any $i$, $1\leqslant i\leqslant|A|$, $p_{\mathbf{w},\mathcal{S}_{i,n}}$ converges uniformly to $p_{\mathbf{w},\mathcal{A}_i}$ on $A^*$ as $n\to\infty$.
\end{lemma}
\begin{proof}
Let $\mathbf{w}\in A^*$ and $p=p_{\mathbf{w},\mathcal{A}_i}$. Assume $\mathbf{w}=\mathbf{u}\mathbf{y}$,
where $|\mathbf{u}\omega|=i-1$ and $|\mathbf{y}|\geqslant0$. After reading $\mathbf{u}$, $\mathcal{S}_{i,n}$
with the same probability $1/(n+1)^{i-1}$ is in one of the $(n+1)^{i-1}$ non-halting states in $\{\mathbf{x}_k\ |\ 1\leqslant k\leqslant n^{|\mathbf{x}|},\ \mathbf{x}\subseteq\mathbf{u}\omega\}$.
Among them, there are $\binom{i-1}{l}n^l$ states corresponding to the states of $\mathcal{A}_i$ at level $l$. So $\mathcal{S}_{i,n}$ has $n^{i-1}$ such states at level $i-1$.
Therefore $p_{\mathbf{w},\mathcal{S}_{i,n}}\geqslant{(\frac{n}{n+1})}^{i-1}p$. Also, $\mathbf{w}$ is rejected with probability
$q_{\mathbf{w},\mathcal{S}_{i,n}}\geqslant{(\frac{n}{n+1})}^{i-1}(1-p)$.
Hence ${(\frac{n}{n+1})}^{i-1}p\leqslant p_{\mathbf{w},\mathcal{S}_{i,n}}\leqslant1-{(\frac{n}{n+1})}^{i-1}(1-p)$.
If $|\mathbf{w}\omega|<i-1$, $p_{\mathbf{w},\mathcal{S}_{i,n}}=p=0$.

Since for any $j$, $0\leqslant j\leqslant|A|$, $\lim\limits_{n\to\infty}(\frac{n}{n+1})^j=1$, $p_{\mathbf{w},\mathcal{S}_{i,n}}$
converges uniformly to $p_{\mathbf{w},\mathcal{A}_i}$.
\qed
\end{proof}
Now it is possible to construct a DH-PRA $\mathcal{S}=(Q,A\cup\{\#,\$\},s,\delta)$, which with the same probability $1/|A|$ executes the automata $\mathcal{S}_{1,n},\dots,\mathcal{S}_{|A|,n}$.
The set of states $Q$ is a disjoint union of $Q_1,...,Q_{|A|}$. Take the initial state $s_i$ of any $\mathcal{S}_{i,n}$ as the initial state $s$.
For any $a\in A\cup\{\$\}$ and $q_1,q_2\in Q_i$, $\delta(q_1,a,q_2)=\delta_i(q_1,a,q_2)$. For any initial states $s_i$ and $s_j$ of $\mathcal{S}_{i,n}$ and $\mathcal{S}_{j,n}$,
$\delta(s_i,\#,s_j)=1/|A|$. For any other state $q$, $\delta(q,\#,q)=1$. So the transition matrices of $\mathcal{S}$ induced by any letter are doubly stochastic.
By Lemma \ref{lemma_pra_unif_conv}, $\mathcal{S}$ recognizes $L$ if $n$ is sufficiently large.

Hence we have established the main result of this section:
\begin{theorem}
\label{theor_sys_consist}
Suppose $L$ is an $\mathcal{R}_1$ language.
If the linear system $\mathfrak{L}$ is consistent, then $L$ can be recognized by a DH-PRA.
\end{theorem}
Therefore, if the linear system $\mathfrak{L}$ is consistent, then $L$ can be recognized
by a MM-BQFA as well. Moreover, since all of the transition matrices of the constructed DH-PRA are also unitary stochastic, by \cite[Theorem 5.2]{GK02} $L$ can be recognized by an EQFA.
\begin{corollary}
\label{cor_not_closed}
The class $\boldsymbol{\mathcal{L}}(\text{MM-BQFA})$ is not closed under union and intersection.
\end{corollary}
\begin{proof}
Consider the $\mathcal{R}_1$ language $L=\{\mathbf{ab,bac}\}$ over alphabet $A=\{a,b,c\}$.
By Corollary \ref{cor_subset}, $L$ can't be recognized by MM-BQFA.

On the other hand, consider the languages $L_1=\{\mathbf{ab}\}$ and $L_2=\{\mathbf{bac}\}$.
Systems $\mathfrak{L}(L_1)$ and $\mathfrak{L}(L_2)$ have the same variables as $\mathfrak{L}(L)$.
The system $\mathfrak{L}(L_1)$ has a solution where $x_a=1/2,y_{ab}=1/2,p_1=1/2,p_2=1$, and all the other variables equal to $0$.
The system $\mathfrak{L}(L_2)$ has a solution where $x_b=1/2,x_{abc}=1/2,p_1=1/2,p_2=1$, and all the other variables equal to $0$.
Therefore by Theorem \ref{theor_sys_consist} the languages $L_1,L_2$ are recognized by MM-BQFA.
Hence $\boldsymbol{\mathcal{L}}(\text{MM-BQFA})$ is not closed under union. The non-closure under intersection is now implied
by closure under complement (Theorem \ref{theor_closed_compl}).
\qed
\end{proof}
\section{Construction of MM-QFA for $\mathcal{R}_1$ languages}\label{section_construct_MM-QFA}
The construction of MM-QFA for $\mathcal{R}_1$ languages has some peculiarities which have to be addressed separately. Specifically, contrary to DH-PRA, EQFA and MM-BQFA,
there exist semilattice languages that MM-QFA do not recognize with probability $1-\epsilon$ \cite[Theorem 5]{AK03} and therefore they can't simulate with the same accepting probabilities the automata
$\mathcal{A}_1,...,\mathcal{A}_{|A|}$ from Section \ref{construct_DH_PRA}. Nevertheless, MM-QFA still recognize any semilattice language and so a modified construction is still possible.

For any $m\in\mathbb{N}$, let $\alpha(m)$ be the least common multiple of $\{1,2,...,m\}$. Also define $\alpha(0)=0$.
Let $O_n$ - $n\times n$ zero matrix. Let $M_n=(m_{rs})=(\frac{1}{n})$ and $U_n=(u_{rs})=\frac{1}{\sqrt{n}}(e^\frac{2\pi i rs}{n})$, where $0\leqslant r,s\leqslant n-1$.
$M_n$ is a doubly stochastic matrix and $U_n$ is a unitary
matrix that represents the discrete Fourier transform. Obtain an $n\times(n-1)$ matrix $V_n$ from $U_n$ by removing in $U_n$ its first column. Let $V_n^*$ - the conjugate transpose of $V_n$.
The following lemma will be useful in the construction of MM-QFA.
\begin{lemma}
\label{lemma_2n-1_unitary}
The $(2n-1)\times(2n-1)$ matrix $H_n=
\begin{pmatrix}
M_n&V_n\\
V_n^*&O_{n-1}
\end{pmatrix}$
is unitary.
\end{lemma}
Suppose $L$ is an $\mathcal{R}_1$ language over alphabet $A$ such that $\mathfrak{L}(L)$ is consistent.
As prescribed in Section \ref{construct_DH_PRA}, we construct the automata $\mathcal{A}_1,...,\mathcal{A}_{|A|}$, which are the components of the probabilistic automaton $\mathcal{A}$
recognizing $L$.

{\em Construction of a MM-QFA.} For any automaton $\mathcal{A}_i$, we construct a sequence of MM-QFA $\mathcal{U}_{i,n}$, where $n\geqslant1$, such that $n^{\alpha(|A|-1)}p_{\mathbf{w},\mathcal{U}_{i,n}}$ converges
uniformly to $p_{\mathbf{w},\mathcal{A}_i}$ on $A^*$ as $n\to\infty$.

An automaton $\mathcal{A}_i=(Q_i,A\cup\{\$\},s_i,\delta_i)$ is used to construct
a MM-QFA $\mathcal{U}_{i,n}=(Q_{i,n},A\cup\{\$\},s_i,\delta_{i,n})$ as described next. If $i>1$, let $c=\frac{\alpha(|A|-1)}{i-1}$, otherwise let $c=0$. Initially $Q_{i,n}$ is empty.
Do the following.
\begin{conditions}
\item For any non-halting state $\mathbf{s}$ at level $j$, $0\leqslant j\leqslant i-1$, supplement $\mathcal{U}_{i,n}$ with non-halting states $\mathbf{s}_k$, where $1\leqslant k\leqslant n^{cj}$.
If $|\mathbf{s}|>0$, new rejecting states $\mathbf{s}_k^\prime$, $1\leqslant k\leqslant n^{cj}$, are added to $\mathcal{U}_{i,n}$ as well.
\item For any non-halting state $\mathbf{s}$ at level $j$, $0\leqslant j<i-1$, supplement $\mathcal{U}_{i,n}$ with rejecting states $(\mathbf{s}\$)_{rej,k}$, where $1\leqslant k\leqslant n^{cj}$.
\item For any non-halting state $\mathbf{s}$ at level $i-1$, accepting state $(\mathbf{s}a)_{acc}$ and rejecting state $(\mathbf{s}a)_{rej}$, where $a\in(A\setminus\mathbf{s})\cup\{\$\}$,
supplement $\mathcal{U}_{i,n}$ with accepting states $(\mathbf{s}a)_{acc,k}$ and rejecting states $(\mathbf{s}a)_{rej,k}$, where $1\leqslant k\leqslant n^{c(i-1)}$.
\end{conditions}
It remains to define the transitions. For any non-halting state $\mathbf{s}$ of $\mathcal{A}_i$ at level $j$, $1\leqslant j\leqslant i-1$, the states in $\{\mathbf{s}_k\}$ are
grouped into $n^{c(j-1)}$ disjoint subsets with $n^c$ states in each, so that any state in $\{\mathbf{s}_k\}$ may be denoted as $\mathbf{s}_{l,m}$, where $1\leqslant l\leqslant n^{c(j-1)}$ and
$1\leqslant m\leqslant n^c$. The states in $\{\mathbf{s}_k^\prime\}$ are grouped in the same way, so that
any state in $\{\mathbf{s}_k^\prime\}$ may be denoted as $\mathbf{s}_{l,m}^\prime$.

For any letter $a$ in $A$, consider all pairs of non-halting states $\mathbf{s,t}$ of $\mathcal{A}_i$ such that $\mathbf{s}\neq\mathbf{t}$ and $\delta_i(\mathbf{s},a,\mathbf{t})=1$.
For any fixed $k$ and any $m$, $1\leqslant m\leqslant n^c$, the transitions induced by $a$ among the states $\mathbf{t}_{k,m},\mathbf{t}_{k,m}^\prime$ and the state $\mathbf{s}_k$ are defined by
the matrix $H_{n^c+1}$; the first row and column is indexed by $\mathbf{s}_k$, the next $n^c$ rows and columns by $\mathbf{t}_{k,m}$, and the last $n^c$ rows and columns by $\mathbf{t}_{k,m}^\prime$.

For any non-halting state $\mathbf{s}$ of $\mathcal{A}_i$ at level $j$, $0\leqslant j<i-1$, $\delta_{i,n}(\mathbf{s}_k,\$,(\mathbf{s}\$)_{rej,k})=1$, $\delta_{i,n}((\mathbf{s}\$)_{rej,k},\$,\mathbf{s}_k)=1$.
For the same $(\mathbf{s}\$)_{rej,k}$ and any other letter $b$ in $A\cup\{\$\}$, define $\delta_{i,n}((\mathbf{s}\$)_{rej,k},b,(\mathbf{s}\$)_{rej,k})=1$.

Consider any non-halting state $\mathbf{s}$ of $\mathcal{A}_i$ at level $i-1$ and $a\in(A\setminus\mathbf{s})\cup\{\$\}$.
Let $r_1=\delta_i(\mathbf{s},a,(\mathbf{s}a)_{acc})$, $r_2=\delta_i(\mathbf{s},a,(\mathbf{s}a)_{rej})$.
If $i=1$, let $u_1=r_1(\frac{1}{n})^{\alpha(|A|-1)}$ and $u_2=1-u_1$. Otherwise, if $i>1$, let $u_1=r_1$, $u_2=r_2$. Note that $u_1+u_2=1$.
The transitions induced by $a$ among $\mathbf{s}_k$, $(\mathbf{s}a)_{acc,k}$, $(\mathbf{s}a)_{rej,k}$
are defined by the matrix
$\begin{pmatrix}
0&0&1\\
\sqrt{u_1}&\sqrt{u_2}&0\\
\sqrt{u_2}&-\sqrt{u_1}&0\\
\end{pmatrix}$.
The first, second and third rows and columns are indexed by
$\mathbf{s}_k$, $(\mathbf{s}a)_{acc,k}$, $(\mathbf{s}a)_{rej,k}$, respectively.
For the same $(\mathbf{s}a)_{acc,k},(\mathbf{s}a)_{rej,k}$ and any
other letter $b$ in $A\cup\{\$\}$, define $\delta_{i,n}((\mathbf{s}a)_{acc,k},b,(\mathbf{s}a)_{acc,k})=\delta_{i,n}((\mathbf{s}a)_{rej,k},b,(\mathbf{s}a)_{rej,k})=1$.

We have defined all the non-zero transitions for $\mathcal{U}_{i,n}$. By construction, the transition matrices induced by any letter $a$ in $A\cup\{\$\}$ are unitary.
\begin{lemma}
\label{lemma_qfa_unif_conv}
For any $i$, $1\leqslant i\leqslant|A|$, $n^{\alpha(|A|-1)}p_{\mathbf{w},\mathcal{U}_{i,n}}$ converges uniformly to $p_{\mathbf{w},\mathcal{A}_i}$ on $A^*$ as $n\to\infty$.
\end{lemma}
\begin{proof}
Let $\mathbf{w}\in A^*$ and $p=p_{\mathbf{w},\mathcal{A}_i}$. If $i=1$, $p_{\mathbf{w},\mathcal{U}_{1,n}}=(\frac{1}{n})^{\alpha(|A|-1)}p$.

Consider the case $i>1$. Assume $\mathbf{w}=\mathbf{u}\mathbf{y}$,
where $|\mathbf{u}\omega|=i-1$ and $|\mathbf{y}|\geqslant0$. After reading $\mathbf{u}$, $\mathcal{U}_{i,n}$ has rejected
the input with probability $1-1/(n^c+1)^{i-1}$ and with the same amplitude $1/(n^c+1)^{i-1}$ is in one of the $(n^c+1)^{i-1}$ non-halting states in
$\{\mathbf{x}_k\ |\ 1\leqslant k\leqslant n^{c|\mathbf{x}|},\ \mathbf{x}\subseteq\mathbf{u}\omega\}$.
Among them, there are $\binom{i-1}{l}n^{cl}$ states corresponding to the states of $\mathcal{A}_i$ at level $l$. So $\mathcal{U}_{i,n}$ has $n^{c(i-1)}$ such states at level $i-1$.
Therefore $p_{\mathbf{w},\mathcal{U}_{i,n}}\geqslant{(\frac{n^c}{(n^c+1)^2})}^{i-1}p$. Also, $\mathbf{w}$ is rejected with probability
$q_{\mathbf{w},\mathcal{U}_{i,n}}\geqslant{(\frac{n^c}{(n^c+1)^2})}^{i-1}(1-p)+1-(\frac{1}{n^c+1})^{i-1}$.
Hence ${(\frac{n^c}{(n^c+1)^2})}^{i-1}p\leqslant p_{\mathbf{w},\mathcal{U}_{i,n}}\leqslant{(\frac{n^c}{(n^c+1)^2})}^{i-1}p+(\frac{1}{n^c+1})^{i-1}-{(\frac{n^c}{(n^c+1)^2})}^{i-1}$.
Note that $n^{\alpha(|A|-1)}=n^{c(i-1)}$, therefore
${(\frac{n^c}{n^c+1})}^{2(i-1)}p\leqslant n^{\alpha(|A|-1)}p_{\mathbf{w},\mathcal{U}_{i,n}}\leqslant{(\frac{n^c}{n^c+1})}^{2(i-1)}p+(\frac{n^c}{n^c+1})^{i-1}-{(\frac{n^c}{n^c+1})}^{2(i-1)}$.
If $|\mathbf{w}\omega|<i-1$, $p_{\mathbf{w},\mathcal{U}_{i,n}}=p=0$.

Since $\lim\limits_{n\to\infty}(\frac{n^c}{n^c+1})^{i-1}=1$, $n^{\alpha(|A|-1)}p_{\mathbf{w},\mathcal{U}_{i,n}}$
converges uniformly to $p_{\mathbf{w},\mathcal{A}_i}$.
\qed
\end{proof}
Construct a MM-QFA $\mathcal{U}_n=(Q,A\cup\{\#,\$\},s,\delta)$, which with the same amplitude $1/\sqrt{|A|}$ executes the automata $\mathcal{U}_{1,n},\dots,\mathcal{U}_{|A|,n}$.
The set of states $Q$ is a disjoint union of $Q_1,...,Q_{|A|}$. Take the initial state $s_i$ of any $\mathcal{U}_{i,n}$ as the initial state $s$.
For any $a\in A\cup\{\$\}$ and $q_1,q_2\in Q_i$, $\delta(q_1,a,q_2)=\delta_i(q_1,a,q_2)$. The transitions induced by initial end-marker $\#$ among the initial states $s_i$ of $\mathcal{U}_{i,n}$, $1\leqslant i\leqslant|A|$,
are defined by the discrete Fourier transform $U_{|A|}$.
For any other state $q$, $\delta(q,\#,q)=1$. So the transition matrices of $\mathcal{U}_n$ induced by any letter are unitary.

We are ready to state the main result of the section.
\begin{theorem}
\label{theor_sys_consist_QFA}
Suppose $L$ is an $\mathcal{R}_1$ language.
If the linear system $\mathfrak{L}$ is consistent, then $L$ can be recognized by a MM-QFA.
\end{theorem}
\begin{proof}
If the linear system $\mathfrak{L}$ is consistent, it is possible to construct the corresponding automaton $\mathcal{U}_n$ from above.
By Lemma \ref{lemma_qfa_unif_conv}, $n^{\alpha(|A|-1)}p_{\mathbf{w},\mathcal{U}_n}$ converges uniformly to $p_{\mathbf{w},\mathcal{A}}$ on $A^*$ as $n\to\infty$.

Take $z=\frac{1}{3}(\mathfrak{P}(p_2)-\mathfrak{P}(p_1))$.
If $n$ is sufficiently large, for any word $\mathbf{u}\in L$ $n^{\alpha(|A|-1)}p_{\mathbf{u},\mathcal{U}_n}\geqslant\mathfrak{P}(p_2)-z$ and for any word $\mathbf{w}\notin L$
$n^{\alpha(|A|-1)}p_{\mathbf{w},\mathcal{U}_n}\leqslant\mathfrak{P}(p_1)+z$.
Hence for all $\mathbf{u}\in L$ $p_{\mathbf{u},\mathcal{U}_n}\geqslant n^{-\alpha(|A|-1)}(\mathfrak{P}(p_2)-z)$ and for all
$\mathbf{w}\notin L$ $p_{\mathbf{w},\mathcal{U}_n}\leqslant n^{-\alpha(|A|-1)}(\mathfrak{P}(p_1)+z)$. So for any $\mathbf{u}\in L$ and $\mathbf{w}\notin L$
$p_{\mathbf{u},\mathcal{U}_n}-p_{\mathbf{w},\mathcal{U}_n}\geqslant n^{-\alpha(|A|-1)}z$

Therefore for a sufficiently large fixed $n$, $\mathcal{U}_n$ recognizes $L$ with bounded error.
\qed
\end{proof}

In summary, we have obtained the following two results:
\begin{theorem}
\label{theor_iff1}
Suppose $L$ is an $\mathcal{R}_1$ language. $L$ can be recognized by MM-QFA if and only if the linear system $\mathfrak{L}(L)$ is consistent.
\end{theorem}
\begin{proof}
By Theorems \ref{theor_sys_not_consist} and \ref{theor_sys_consist_QFA}.
\qed
\end{proof}
\begin{theorem}
\label{theor_iff2}
MM-QFA, DH-PRA, EQFA and MM-BQFA recognize exactly the same $\mathcal{R}_1$ languages.
\end{theorem}
\begin{proof}
By Theorems \ref{theor_sys_not_consist}, \ref{theor_sys_consist} and \ref{theor_sys_consist_QFA}.
\qed
\end{proof}
\section{"Forbidden Constructions"}\label{section_forb_constr}
In \cite[Theorem 4.3]{AV00}, \c Kikusts has proposed "forbidden constructions" for MM-QFA; any regular language whose minimal deterministic finite automaton contains any of these
constructions cannot be recognized by MM-QFA. It is actually implied by Theorem \ref{theor_bist_EJ} that the same is true for MM-BQFA and other "decide-and-halt" models from Table \ref{table_automata}.
Also, by Theorem \ref{theor_MM_BQFA_subset_ER} any language that is recognized by a MM-BQFA is contained in $\boldsymbol{\mathcal{ER}}$.
Therefore it is legitimate to ask whether all the $\boldsymbol{\mathcal{ER}}$ languages that do not contain any of the "forbidden constructions" can be recognized by
MM-BQFA. In this section, we give a {\em negative} answer to this question; we provide an example of an $\mathcal{R}_1$ language that does not contain any
of the "forbidden constructions" and still cannot be recognized by MM-BQFA (and by other "decide-and-halt" models from Table \ref{table_automata}).

First, we need a lemma.
\begin{lemma}
\label{lemma_forb_constr}
An $\mathcal{R}_1$ language $L$ has a "forbidden construction" with $n+1$ levels if and only if there exist $m,n$ and words $\mathbf{w}_i$, $\mathbf{x}_{i,k}$, $1\leqslant i\leqslant2m$, $1\leqslant k\leqslant n$,
such that
\begin{conditions}
\item \label{cond_m_1} $\mathbf{w}_1,...,\mathbf{w}_m\in L;$
\item $\mathbf{w}_{m+1},...,\mathbf{w}_{2m}\notin L;$
\item \label{cond_m_3} for any $i$, $\mathbf{w}_i=\mathbf{x}_{i,1}...\mathbf{x}_{i,n}$;
\item for any $i,k$ $\mathbf{x}_{i,k}=\mathbf{x}_{i,k}\tau$;
\item \label{cond_m_4} for any $i,j,k$, if $1\leqslant k<n$ then $\mathbf{x}_{i,k}\sim_\omega\mathbf{x}_{j,k}$;
\item \label{cond_m_6} for any $i$ $\mathbf{w}_i=\mathbf{w}_i\tau$;
\item \label{cond_m_5} for any $k$ the tuple $(\mathbf{x}_{1,k},...,\mathbf{x}_{m,k})$ is a permutation of $(\mathbf{x}_{m+1,k},...,\mathbf{x}_{2m,k})$.
\end{conditions}
\end{lemma}
\begin{proof}
Assume an $\mathcal{R}_1$ language $L$ has a "forbidden construction" of $n+1$ levels. Let $l_k$ - the number of different labels (words) for transitions between levels $k$ and $k+1$.
Let $\mathbf{z}_{1,k},...,\mathbf{z}_{l_k,k}$ - the words labeling the transitions from the states at level $k$ to the states at level $k+1$.
If $k<n$ there exist transitions labeled $\mathbf{z}_{1,k}^\prime,...,\mathbf{z}_{l_k,k}^\prime$ between the states at level $k$ to the states at level $k+1$
such that $\mathbf{z}_{1,k}^\prime\sim_\omega...\sim_\omega\mathbf{z}_{l_k,k}^\prime$ and for all $i$ $\mathbf{z}_{i,k}^\prime\tau=\mathbf{z}_{i,k}^\prime$.
If $k=n$ there exist transitions labeled
$\mathbf{z}_{1,n}^\prime,...,\mathbf{z}_{l_n,n}^\prime$ between the states at level $n$ to the states at level $n+1$ such that for all $i$ $\mathbf{z}_{i,n}^\prime\tau=\mathbf{z}_{i,n}^\prime$.
The states at level $n+1$ are a disjoint union of the sets $D_{1,n},...,D_{l_n,n}$. Therefore the last level
has $m$ accepting and $m$ rejecting states, where $m>0$. Hence there are $m$ words $\mathbf{w}_i^\prime$ in $L$, $1\leqslant i\leqslant m$, and $m$ words $\mathbf{w}_j^\prime$ not in $L$, $m+1\leqslant j\leqslant2m$.
For any $i$, $1\leqslant i\leqslant2m$, $\mathbf{w}_i^\prime=\mathbf{x}_{i,1}^\prime...\mathbf{x}_{i,n}^\prime$, where $\mathbf{x}_{i,k}^\prime$ is equal to some label
$\mathbf{z}_{s,k}^\prime$, where $1\leqslant s\leqslant l_k$. So the words $\mathbf{w}_i^\prime$ satisfy the conditions (\ref{cond_m_1}-\ref{cond_m_4}).
Consider the set $D_{s,k}$. Since it has the same number of accepting and rejecting states, $\mathbf{z}_{s,k}^\prime$ occurs the same number of times in the tuples
$(\mathbf{x}_{1,k}^\prime,...,\mathbf{x}_{m,k}^\prime)$ and $(\mathbf{x}_{m+1,k}^\prime,...,\mathbf{x}_{2m,k}^\prime)$. This implies the condition (\ref{cond_m_5}).
For all $i$, let $\mathbf{w}_i=\mathbf{w}_i^\prime\tau$. Since for all $i,j$ and for all $k$ less than $n$ $\mathbf{x}_{i,k}^\prime\sim_\omega\mathbf{x}_{j,k}^\prime$, the application of $\tau$ to the words $\mathbf{w}_i^\prime$
will delete for any $k$ the same letters in $\mathbf{x}_{1,k}^\prime,...,\mathbf{x}_{2m,k}^\prime$, thus producing words $\mathbf{x}_{1,k},...,\mathbf{x}_{2m,k}$. So for any $i$
$\mathbf{w}_i=\mathbf{x}_{i,1}...\mathbf{x}_{i,n}$. The words $\mathbf{w}_1,...,\mathbf{w}_{2m}$ satisfy all the conditions (\ref{cond_m_1}-\ref{cond_m_5}).

Now suppose the language $L$ satisfies the conditions (\ref{cond_m_1}-\ref{cond_m_5}). It is possible to construct a following "forbidden construction".
Level $1$ consists of a state $q_1$ and the words $\mathbf{x}_{1,1},...,\mathbf{x}_{2m,1}$. Level $2$ consists of states $q_{1,2},...,q_{2m,2}$, such that for any $i,j$ $q_1\mathbf{x}_{i,1}=q_{i,2}$,
$q_{i,2}\mathbf{x}_{j,1}=q_{i,2}$.
Level $2$ also has the words $\mathbf{x}_{1,2},...,\mathbf{x}_{2m,2}$. Level $k$, $3\leqslant k\leqslant n$, consists of states $q_{1,k},...,q_{2m,k}$, such that for any $i,j$ $q_{i,k-1}\mathbf{x}_{i,k-1}=q_{i,k}$,
$q_{i,k}\mathbf{x}_{j,k-1}=q_{i,k}$.
Level $k$ also has the words $\mathbf{x}_{1,k},...,\mathbf{x}_{2m,k}$. Level $n+1$ consists of states $q_{1,n+1},...,q_{2m,n+1}$, such that for any $i$ $q_{i,n}\mathbf{x}_{i,n}=q_{i,n+1}$.
(Within a "forbidden construction", two states may represent the same state in a minimal deterministic automaton, so it is legible to have the same label
in two transitions outgoing a single state.)
The states $q_{1,n+1},...,q_{m,n+1}$ are accepting and the states $q_{m+1,n+1},...,q_{2m,n+1}$ are rejecting.
\qed
\end{proof}
\begin{theorem}
\label{theor_forb_constr}
There exists an $\boldsymbol{\mathcal{ER}}$ language that does not contain any of the "forbidden constructions" and still cannot be recognized by MM-BQFA.
\end{theorem}
\begin{proof}
Consider an $\mathcal{R}_1$ language $L=\{\mathbf{aedbc}$, $\mathbf{beca}$, $\mathbf{beda}$, $\mathbf{bedac}$, $\mathbf{eacb}$, $\mathbf{eacbd}$, $\mathbf{eadbc}$, $\mathbf{ebca}\}$
over alphabet $A=\{a,b,c,d,e\}$.
Among others, the system $\mathfrak{L}$ has the following inequalities:
$$
\begin{array}{lclclclclclcrcl}
\mathfrak{L}(\mathbf{aedbc})&=&x_0+x_a&+&x_{ae}&+&x_{aed}&+&x_{adeb}&+&x_{abdec}&+&y_{abcde}&\geqslant&p_2\\
\mathfrak{L}(\mathbf{beda})&=&x_0+x_b&+&x_{be}&+&x_{bed}&+&x_{bdea}&+&y_{abde}&&&\geqslant&p_2\\
\mathfrak{L}(\mathbf{eacbd})&=&x_0+x_e&+&x_{ea}&+&x_{aec}&+&x_{aceb}&+&x_{abced}&+&y_{abcde}&\geqslant&p_2\\
\mathfrak{L}(\mathbf{ebca})&=&x_0+x_e&+&x_{eb}&+&x_{bec}&+&x_{bcea}&+&y_{abce}&&&\geqslant&p_2\\
\mathfrak{L}(\mathbf{aecb})&=&x_0+x_a&+&x_{ae}&+&x_{aec}&+&x_{aceb}&+&y_{abce}&&&\leqslant&p_1\\
\mathfrak{L}(\mathbf{becad})&=&x_0+x_b&+&x_{be}&+&x_{bec}&+&x_{bcea}&+&x_{abced}&+&y_{abcde}&\leqslant&p_1\\
\mathfrak{L}(\mathbf{eadb})&=&x_0+x_e&+&x_{ea}&+&x_{aed}&+&x_{adeb}&+&y_{abde}&&&\leqslant&p_1\\
\mathfrak{L}(\mathbf{ebdac})&=&x_0+x_e&+&x_{eb}&+&x_{bed}&+&x_{bdea}&+&x_{abdec}&+&y_{abcde}&\leqslant&p_1\\
&&&&&&&&&&&&p_1&<&p_2
\end{array}
$$
$
\begin{array}{rclrclrcl}
\text{Let\qquad}a_1&=&x_0+x_a+x_{ae},&\qquad b_1&=&x_{aed}+x_{adeb},&\qquad c_1&=&x_{abdec}+y_{abcde},\\
a_2&=&x_0+x_b+x_{be},&\qquad b_2&=&x_{bed}+x_{bdea},&\qquad c_2&=&y_{abde},\\
a_3&=&x_0+x_e+x_{ea},&\qquad b_3&=&x_{aec}+x_{aceb},&\qquad c_3&=&x_{abced}+y_{abcde},\\
a_4&=&x_0+x_e+x_{eb},&\qquad b_4&=&x_{bec}+x_{bcea},&\qquad c_4&=&y_{abce}.\\
\end{array}
$

$$
\begin{array}{rclrclr}
\text{We obtain inequalities\quad}a_1+b_1+c_1&\geqslant&p_2,&\qquad a_1+b_3+c_4&\leqslant&p_1,&\qquad p_1<p_2,\\
a_2+b_2+c_2&\geqslant&p_2,&\qquad a_2+b_4+c_3&\leqslant&p_1,\\
a_3+b_3+c_3&\geqslant&p_2,&\qquad a_3+b_1+c_2&\leqslant&p_1,\\
a_4+b_4+c_4&\geqslant&p_2,&\qquad a_4+b_2+c_1&\leqslant&p_1,
\end{array}
$$
which define a system that is not consistent. Hence $\mathfrak{L}$ is not consistent as well.
So by Theorem \ref{theor_sys_not_consist} $L$ cannot be recognized by MM-BQFA.

Let us check if $L$ contains any of the "forbidden constructions" from \cite[Theorem 4.3]{AV00}.
Since $L$ is $\mathcal{R}$-trivial idempotent and $|A|=5$, if $L$ contains some "forbidden construction",
by Lemma \ref{lemma_forb_constr}, it also must contain a construction with number of levels not larger than $6$.
Therefore it remains to check the conditions of Lemma \ref{lemma_forb_constr} against constructions with number of levels equal to $3,4,5$ and $6$.
In case of $3$ levels, it is sufficient to verify that any subset of  $\{\mathbf{aedbc}$, $\mathbf{beca}$, $\mathbf{beda}$, $\mathbf{bedac}$, $\mathbf{eacb}$, $\mathbf{eacbd}$, $\mathbf{eadbc}$, $\mathbf{ebca}\}$
with at least two elements does not form the words $\mathbf{w}_1,...,\mathbf{w}_m$ satisfying all the conditions of Lemma \ref{lemma_forb_constr}. Actually, it is sufficient to check only the subsets
of $\{\mathbf{aedbc}$, $\mathbf{eacb}$, $\mathbf{eacbd}$, $\mathbf{eadbc}\}$, $\{\mathbf{beca}$, $\mathbf{beda}$, $\mathbf{bedac}$, $\mathbf{ebca}\}$,
$\{\mathbf{aedbc}$, $\mathbf{beda}$, $\mathbf{bedac}$, $\mathbf{eadbc}\}$ and $\{\mathbf{beca}$, $\mathbf{eacb}$, $\mathbf{eacbd}$, $\mathbf{ebca}\}$.
None of these subsets satisfy the conditions of the lemma.
The cases with $4,5$ and $6$ levels are checked in the same way. So $L$ does not contain any of the "forbidden constructions".
\qed
\end{proof}

\end{document}